\documentclass[11pt]{article}
\usepackage{authblk} 
\usepackage[margin=1in]{geometry}
\usepackage{graphicx}
\usepackage{amssymb}
\usepackage{amscd}
\usepackage{textcomp}

\usepackage{amsthm}
\usepackage{hyperref}
\usepackage{coqdoc}
\usepackage{url}

\makeatletter
\@ifundefined{coqdef}{
  \newcommand{\coqdef}[3]{%
    \phantomsection%
    \hypertarget{coq:#1}{#3}%
  }
}{}
\makeatother

\urldef{\mailsa}\path|{patrick.barlatier}, richard.dapoigny@univ-savoie.fr|    

\newtheorem{defi}{Definition}
\newtheorem{prop}{Proposition}

\newtheorem{ax}{Axiom}

\title{A Coq-based Axiomatization of Tarski's Mereogeometry}
\author[1]{Patrick Barlatier}
\author[1]{Richard Dapoigny}

\affil[1]{University of Savoie Mont-Blanc, LISTIC}
\affil[ ]{\texttt{patrick.barlatier@univ-smb.fr}}
\affil[ ]{\texttt{richard.dapoigny@univ-smb.fr}}

\date{}

\begin{document}
\maketitle

\begin{abstract}
During the last decade, the domain of Qualitative Spatial Reasoning, has known a renewal of interest for mereogeometry, a theory that has been initiated by Tarski. Mereogeometry relies on mereology, the Le{\'s}niewski's theory of parts and wholes that is further extended with geometrical primitives and appropriate definitions. However, most approaches (i) depart from the original Le{\'s}niewski's mereology which does not assume usual sets as a basis, (ii) restrict the logical power of mereology to a mere theory of part-whole relations and (iii) require the introduction of a connection relation. Moreover, the seminal paper of Tarki shows up unclear foundations and we argue that mereogeometry as it is introduced by Tarski, can be more suited to extend the whole theory of Le{\'s}niewski. For that purpose, we investigate a type-theoretical representation of space more closely related with the original ideas of Le{\'s}niewski and expressed with the Coq language. We show that (i) it can be given a more clear foundation, (ii) it can be based on three axioms instead of four and (iii) it can serve as a basis for spatial reasoning with full compliance with Le{\'s}niewski's systems.
\end{abstract}
\begin{center}
\begin{minipage}{0.85\textwidth}
\footnotesize
\textbf{Note.} This is the author's accepted manuscript of the paper:
P. Barlatier and R. Dapoigny “A Coq-Based Axiomatization of Tarski’s Mereogeometry”
COSIT 2015, in Lecture Notes in Computer Science.
The final authenticated version is available at SpringerLink.		
\end{minipage}
\end{center}
\section{Introduction}
% QSR needs mereotopology or mereogeometry
In Knowledge Representation and Reasoning (KRR), and especially in the sub-domain of Qualitative Spatial Reasoning (QSR) the very nature of topology and its relation to how humans perceive space stems from what is called mereotopology \cite{Rand92,Asher95,Varzi96,Hahmann10}. The term mereotopology refers (i) to mereology, i.e., the theory of parts and wholes (see e.g., \cite{Simon87} for a complete analysis) and (ii) to the addition of topological relation(s) to mereology (e.g., the \textgravedbl Connected\textacutedbl relation) in order to get sufficient expressiveness for reasoning about space. Notice that mereotopology is general enough since it can be applied in many other fields where the spatial character is not the primary purpose. An alternative approach for the modeling of spatial regions which relies on geometrical primitives is known as mereogeometry \cite{Tarski56a}. It has known some gain of interest during the last decade \cite{Benn01,Grusz08,Borgo10,Betti12} and appears as a promising research area. Mereo-geometrical relations are required to be invariant to the strength of the type of geometry, e.g., for affine geometry they should be invariant under affine transformations. Mereogeometry has been applied mainly in the area of control tasks for mobile robotics \cite{Kuipers87,Korten98} and in the semantics of spatial prepositions \cite{Aurnag95}.
% motivations for region-based theories and pb of mereo
\par The development of such region-based theories can be seen as an appealing alternative to set theory, point-set topology and Euclidian geometry. The quest for such theories of space are primarily motivated by human cognition, i.e., how humans perceive their spatial environment. It is also well-known that these region-based theories are able to draw topological or mereological conclusions even in the absence of precise data. Furthermore, they bridge the gap between low-level and high-level representations of space by providing a way to understand the nature of points, e.g., offering a structure that is not evident in Euclidian geometry \cite{Eschenbach95}. 
% focus sur mereogeometry
\par Region-based theories using topological properties have an expressive power which is much more restricted than point-based geometry. 
From that perspective, mereogeometry and more especially, Tarski's mereogeometry, appears as a powerful alternative. However, it has not be fully formalized, and several authors have provided a fully formal system based on Tarski'work, either using a first-order language \cite{Borgo96} or set theory \cite{Grusz08}, or using parthood together with a sphere predicate \cite{Benn01}. As far as we know, the totality of these approaches consider mereology as the theory of parts and wholes and restrict the mereological contribution to the theory, as a single \textit{part-of} relation. This assumption leads to many difficulties such as the addition of the so-called Weak/Strong Supplementation principle whose assumption is highly debatable \cite{Simon87,Varzi96}. Furthermore, mereogeometry is not compelled to provide a connection primitive \cite{Tarski56a,Benn01,Grusz08}. 
\par We observe that all the mereogeometrical theories developed so far (i) depart from the classic mereology of Le{\'s}niewski which does not assume usual sets as a basis, (ii) perceive mereology as a mere theory of part-whole relations\footnote{Most descriptions suffer from the misunderstanding that mereology is formally nothing more than a particular elementary theory of partial ordering.} and (iii) restrict the underlying logic to first-order logic whereas the original theory of Le{\'s}niewski is clearly higher-order. We argue that mereogeometry as it is introduced by Tarski, appears more suited to extend the foundations of Le{\'s}niewski's mereology for building a sound theory of space.
\par The first objective of the paper is to motivate for a mereogeometry based on Le{\'s}niewski's work with the purpose of providing a well-founded region-based theory of space \cite{Hahmann10}. The second objective is to develop a type-theoretical theory of space more closely related with the original ideas of Le{\'s}niewski \cite{Lesniewski16}. Finally, besides expressing Le{\'s}niewski's mereology with the Coq language \cite{Bertot04}, we are able to provide a version of Tarski's mereogeometry in Coq. After a short presentation of motivations in section 2, we describe the basis of Le{\'s}niewski's mereology in section 3. The type-theoretical account of  Le{\'s}niewski's mereology is summarized in section 4. Section 5 discusses the foundations of mereogeometry which (i) extends the type-theoretical mereology and (ii) builds upon Tarski's work. 
\section{Motivations for a Le{\'s}niewski's approach to mereogeometry }
We first present some arguments which advocate for using mereogeometry as an appealing alternative to mereotopology. Besides, we show how a detailed analysis of Tarski's papers will motivate the Le{\'s}niewskian approach as a sound basis for mereogeometry.
\par In mereotopology as well as in mereogeometry, the common theoretical core is the so-called mereology basis. However, in the quasi-totality of approaches the term mereology has little connection (if any) with the Le{\'s}niewski's Mereology. The notable exception is the mereogeometry of Tarski \cite{Tarski29,Tarski56a}. The strong point for using Le{\'s}niewski's mereology stands in the fact that it can be built from algebraic structures, and more precisely from a quasi boolean algebra (i.e., without the zero element). It was stated by Tarski \cite{Tarski56b} and further proved by Clay \cite{Clay74} that every mereological structure can be transformed into complete Boolean lattice by adding zero element (its non-existence is a consequence of axioms for mereology) and conversely, that every complete Boolean lattice can be turned into a mereological structure by deleting the zero element. Furthermore, a major difficulty we see in mereotopological theories stems from the addition of a purely topological primitive (i.e., connection) to a mereological framework. This commitment results in hybrid systems whose theoretical foundations are unclear. As underlined in \cite{Benn01}, mereotopological theories are limited to describe topological properties and therefore, their expressive power is more restricted than point-based geometry. Another benefits of using mereogeometry is its ability to reconstruct points from region-based primitives \cite{Gerla95}. For instance, in Tarski's mereogeometry, the domain of discourse is classified into spheres and thus one can recover points from spheres i.e., metrical ones. For these reasons, we argue for a mereogeometrical framework using the Le{\'s}niewski's mereological deductive basis. The best candidate for such a framework is the geometry of solids first proposed in \cite{Tarski29} and modified in \cite{Tarski56a}.
\par However, a detailed analysis of Tarski's papers shows up some unclear methodological aspects \cite{Betti12}. The authors point out that Tarski's theory suffers from (i) some methodological divergences from Le{\'s}niewski's mereology and (ii) a free mixing of Le{\'s}niewski's mereology and the type-theoretical approach of Whitehead and Russel's \textit{Principia mathematica}. Point (i) refers to the facts that in Le{\'s}niewski's theories, axioms precede definitions and that axioms should not contain defined notions \cite{Betti13}. Whereas in the early paper \cite{Tarski29} Tarski assumes the availability of Le{\'s}niewski's mereology, in the second paper \cite{Tarski56a} he provides a minimal mereology based on three definitions and two axioms. As advocated in \cite{Rickey77}, the axiomatization of mereology given by Tarski is deficient in that one of the Le{\'s}niewski's axioms of mereology is not provable from his axioms. By adopting the original version of Le{\'s}niewski as axiom system, we are able to prove that this assertion is irrelevant here. However, while the proposed version fails to come up to the aesthetic rules of \textgravedbl well-constructed axiom system\textacutedbl, it has the virtue of intuitive simplicity and it expresses faithfully Le{\'s}niewski's mereological principles. Point (ii) requires the clarification of two aspects: (a) the exact semantics behind the notions of \textgravedbl class\textacutedbl and \textgravedbl sum\textacutedbl and (b) the relation between the domain of discourse and the range of quantifiers. To make (a) precise, a footnote in the second paper of Tarski \cite{Tarski56a} claims that the sense in which class should be interpreted is Russellian. While the approach of Bennett \cite{Benn01} develops a region-based version of Tarski's theory by restricting the mereology to first-order axioms without involving set-theory, it results in an undecidable model. Alternatively, the analysis detailed in \cite{Grusz08} relies on an algebraic framework which supports higher-order logic but uses set theory which is incoherent if one expects to work with Le{\'s}niewski's mereology. Departing from these approaches, the contributions of the paper turn out (i) to formalize Tarski's geometry of solids w.r.t. Le{\'s}niewski's mereology in a strict methodological sense and (ii) to provide a type-theoretical account for the whole theory based on the Calculus of Inductive Constructions.
\section{Le{\'s}niewski's Mereology}
\subsection{Foundations}
In the early 20th century, S. Le{\'s}niewski first proposed a higher-order logical theory (called protothetic) based on a single axiom, the equivalence axiom. This logical system was the support of a theory called mereology, whose purpose was the description of the world with collective classes and mereological relations such as the so-called \textit{part-of} relation. While set theory relies on the opposition between element and set, mereology is rather based on the opposition beteen part and whole. Mereology primarily assumes the distinction between the distributive and collective interpretations of a class. A distributive class (a usual set) is the extension of a concept such as for example, the solar system, $P = \lbrace Mercury, Venus, Earth, Jupiter, Mars, Saturn, Neptune, Pluto \rbrace$ in which $P$ is a distributive class which contains nine elements and nothing else. By contrast, the notion of collective class involves both the previous planets and a lot of entities such as the moon, the oceans, water, the fishes in it, the red Jupiter spot, the rings of Saturn and the like. It should also be noted that Le{\'s}niewski's mereology assumes the interplay between distributive and collective classes. The ontology introduces names, where name is the distributive notion, while the mereology relies on classes, where class is the collective notion. The theory also admits two kinds of variables which are either of the type propositional and belong to the semantic category $S$ or to the category of names $N$. Any name may refer to a singular name, a plural name or the empty name. Notice that the introduction of plural names gives mereology an expressive power that goes beyond the capabilities of first-order logic (e.g., predicates may be introduced which have plural subjects). The logical system relies on the equivalence relation which allows to introduce new symbols as names of arbitrary elements of any proposition. Le{\'s}niewski's idea was to construct entirely new foundation for mathematics. For that purpose, its theory consists of three subsystems:
\begin{itemize}
\item prototethics which corresponds roughly to a higher-order propositional calculus introduces the equivalence, the category $S$ and extensionality,
\item ontology referred to as Le{\'s}niewski's Ontology (LO) which includes protothetic, could be described as a theory of particulars. It introduces the copula \textgravedbl $\varepsilon$\textacutedbl (distinct from the set-theoretic $\in$) and the category $N$,
\item mereology (which includes LO) whose axioms are not deducible from logical principles is rather a theory which introduces collective classes, name-forming functor \textgravedbl class of\textacutedbl, \textgravedbl part of\textacutedbl and \textgravedbl element of\textacutedbl relations together with their properties.
\end{itemize}
In the following upper case letters will denote singular names while lowercase will refer to plural names. Le{\'s}niewski's mereology can be presented in many settings depending on the adopted set of primitive terms. We refer here to the original system (1916). Notice that different versions of mereology have been proposed, but it has been proved that each of these versions are inferentially equivalent to the original system.
\par \textbf{Protothetic} \par Protothetic is a logically rigorous system, a propositional calculus with quantifiers and semantic categories developed for supporting LO, which itself supports mereology. The construction of computable systems of protothetic (referred as $S_1$ to $S_5$) allowed Le{\'s}niewski to prove that $S_5$ is a consistent and complete system \cite{Lesniewski38,Slupecki53}. In protothetic, a single category $S$ is used as the universe of propositions. It is a higher-order two-valued logic based on the equivalence relation denoted $\equiv$ between terms which can be propositions or propositional expressions (e.g., a propositional function having two propositional arguments such as $S/SS$\footnote{The two argument categories are on the right side and the resulting category on the left side.}). Without getting into the details, the formalization of protothetic obeys a succession of refinements that could broadly divided in two steps: (i) a quantified equivalential calculus and (ii) a logic of bivalent propositions which should support quantification on any variable whatever its category, and should admit the law of extensionality for propositions. The extensionality principle for propositional types is written as:
\begin{eqnarray*}
\forall p \: q \: f,  (p \equiv q) \: \equiv \: (f(p) \equiv f(q))
\end{eqnarray*}
where $p$, $q$ and $f$ denote propositional variables.
\par \textbf{Le{\'s}niewski's Ontology} \par LO can be perceived as an extension of the traditional Aristotelian formal logic and interpreted as a theory a certain kind of general principles of existence. Given two names, $a$ and $b$, the sentence $a \varepsilon b$ is true iff $a$ denotes exactly one object, and this object is named by $b$. The distributive meaning of classes is captured by $\varepsilon$ since \textgravedbl A is an element of the extension of the objects a\textacutedbl, is identical to \textgravedbl A is a\textacutedbl, i.e., $ A \varepsilon a$. Intuitively it means that one object (individual) falls under the scope of a collective notion (a sum of individuals). A single axiom states the properties of the copula $\varepsilon$:
\begin{eqnarray*}
\forall A a, \: A \varepsilon a & \equiv & ((\exists B, B \varepsilon A) \wedge  (\forall C D, (C \varepsilon A \wedge D \varepsilon A) \rightarrow C \varepsilon D) \wedge (\forall C, C \varepsilon A  \rightarrow C \varepsilon  a))
\end{eqnarray*}
The first conjunction of the right side of the equivalence prevents $A$ from being an empty name, the second conjunction states the uniqueness of A while the last conjunction refers to a kind of convergence (anything which is $A$ is also an $a$). The only functor of LO provides the meaning of $\varepsilon$ via a single axiom and rules of the system. Many definitions are given on the form of equivalences, that is using primitives of the language rather than meta-level primitives which do not belong to the language. We only describe definitions which are of interest for the purpose of the paper. The first one defines the $singular\_equality$ which holds when two singular names denote the same object.
\begin{eqnarray*}
\forall A \: B,  singular\_equality \: A \: B \equiv (A \: \varepsilon \: B \: \wedge \: B \: \varepsilon \: A)
\end{eqnarray*}
This relation is symmetric and transitive, but not reflexive. The second one is the $weak\_equality$ between plural names (symmetric, reflexive and transitive), i.e., an equivalence relation. 
\begin{eqnarray*}
\forall a \: b,  weak\_equality \: a \: b \equiv \forall A, (A \: \varepsilon \: a) \equiv (A \: \varepsilon \: b)
\end{eqnarray*}
The extensionality principle is involved to prove the following theorem of LO where $A$, $B$, $C$ are members of the category $N$, while $\phi$ belongs to $N/N$:
\begin{eqnarray*}
\forall A \: B \: C \: \phi, \; A \: \varepsilon \: (\phi \: B) \: \wedge singular\_equality \: B \: C \; \rightarrow \; A \: \varepsilon \: (\phi \: C)
\end{eqnarray*}
\par \textbf{Mereology} \par Mereology is an axiomatic theory of parts built on LO which relies on the concept of (collective) classes. For that purpose, it can be formalized in terms of different primitives. The most usual formalization introduces a mereological element called part (\textsf{pt}), as a primitive. Mereology is developed on a minimal collection of axioms and new primitive functors are defined whose the most important are that of (collective) class, denoted \textsf{Kl} and that of element of a class referred to as \textsf{el}. From now, it should be clear that the notions of \textgravedbl collective class\textacutedbl, \textgravedbl mereological sum\textacutedbl and \textsf{Kl} are equivalent. Besides, it has been demonstrated that Le{\'s}niewski's mereology is consistent \cite{Clay68,Lejewski69}. 
\subsection{Contribution of Le{\'s}niewski's Mereology}
Representing all the Le{\'s}niewski's apparatus in Coq will add some value to the current work. The first benefit comes from the underlying higher logic (protothetic) which is coherent with the choice of a higher-order type theory on which Coq relies. The higher-order framework of Coq offers much more expressiveness than first-order provers. For example, relations can be composed at the meta level from existing relations. The usual critique concerns the lack of automation inherent to higher-order provers, however, Coq comes with a proof-search mechanism and with several decision procedures that enable the system to automatically synthesize lumps of proof. 
\par The second benefit relates to the conceptualization. First, using Le{\'s}niewski's mereology simplifies Tarski's approach to mereogeometry and clarifies many semantic issues (see section \ref{Mereogeometry}). Using LO is a way to provide sound and simple ontological foundation for spatial reasoning. Second, the theory is able to aggregate in a single framework distributive and collective classes giving rise to an expressive and coherent theory. For example, instead of representing the relation of \textgravedbl proper-part of\textacutedbl as a relation between two names $x$ and $y$, Le{\'s}niewski's mereology provides a more expressive relation \cite{Simon87} using the two-place functor $\varepsilon$ and the one-place functor $pt$ with: $x \: \varepsilon \: pt \: y$ resulting in a stronger system. Not only we get more information but, we also avoid to eliminate the notion of distributive class \cite{Clay74} and consequently the interplay between the two notions of class (distributive and collective).
\par An alternative to mereogeometry is the so-called mereotopology \cite{Varzi96}. However, there is some foundational difficulty here. On the one hand, the original mereology of Le{\'s}niewski is rather unknown\footnote{This is partly due to the destruction of most of his work during the second world war and to the difficulty to assess protothetic.} while on the other hand, most works view his work as a mere theory of parts and mix it inside a set-based framework. In fact, Le{\'s}niewski's theory is incompatible with a set-based approach and it is one of the underlying objectives of our work to provide a coherent theory independent from set theory. Mereo-geometrical definitions (see section \ref{Mereogeometry}) rely on user's requests that can be solved using lemmas and tactics in a quasi-automatic mode. The library proved in \cite{Dapoi15} uses more than a hundred theorems that are all available for reasoning in mereogeometry. 
While mereology-based axiomatizations of Tarski’s mereogeometry was given previously, the paper has a number of new contributions: (i) it incorporates Lesniewski's mereology (with its protothetic and ontology components) together with Tarski's mereogeometry in a coherent system, (ii) it is fully formalized and (iii) it is computer verified.

%RCC-8 (Region Connection Calculus) developed by A. Cohn and D. Randell \cite{Rand92}, which is based on the calculus of individuals \cite{Clark81}. From eight basic relations, composition tables are at the core of reasoning with spatial relations. However, as underlined in \cite{Bouzy01}, the information contained in cells of the composition tables is rather weak for knowledge-based concrete applications. Mereo-geometrical definitions (see section \ref{Mereogeometry}) constrast with this approach and rather rely on user's requests that can be solved using lemmas and tactics in a quasi-automatic mode. Furthermore, the ability to reason with the properties of these relations in higher-order logic is another substantial improvement. The definition of the interior of a set comes to the introduction of a \textgravedbl quasi-topology\textacutedbl which has several drawbacks (e.g., idempotence fails). No such difficulty occurs in the present approach where (collective) classes as well as collections are idempotent (the class of a class is a class).
%Furthermore, due to the existence of the universal region, an extensional interpretation is not compatible with the RCC theory \cite{Li03}. Assuming the existence of a universe in Le{\'s}niewski's Mereology does not affect the extensional interpretation for functors (i.e., relations in set-based theories).

\section{The Type-theoretical Mereology}\label{TypeMereo}
\subsection{Using Coq as a Logical Foundation}\label{Unif-frame}
Tarski has pointed out that propositional connectives can be defined in terms of logical equivalence and quantifiers in the context of higher-order logic \cite{Tarski23}. Then, Quine further showed how quantifiers and connectives can be expressed in terms of equality and the abstraction operator in the context of Church's type theory \cite{Quine56}. These results led Henkin to define the propositional connectives and quantifiers as they appear in protothetic, in type theory having at least function application, function abstraction, and equality \cite{Henkin63}. Versions of dependent type theory in computer science such as the Calculus of Inductive Constructions (CIC) are largely inspired from Church-style (constructive) logics. As a consequence, there is a semantic connection between the logical content of CIC and protothetic since the former can be seen as an extension of Church's type theory which itself is able to represent propositional connectives and quantifiers of protothetic. As underlined in \cite{Appel04,Bertot08}, there is a significant advantage for using a dependently typed functional language such as CIC and its implementation with the Coq language, mainly because its programming language combines powerful logical capabilities and reasonable expressive power. 
\par We have to prove that the system of axioms and rules of inference governing sentential connectives and quantifiers of protothetic is expressible in Coq. The qualitative measure suggested here is relative to the system $S5$ of protothetic and the CIC which is, with some variants, the one implemented in the theorem prover Coq. 
\begin{prop}
CIC is at least as expressive as protothetic.
\end{prop}
\begin{proof} The first step is to prove that CIC is at least as expressive as Church Simple Type Theory (STT)\cite{Church40}. Whereas STT has variables ranging over functions, together with binders for them, System $F$\footnote{Known as Girard--Reynolds polymorphic lambda calculus. It extends STT by the introduction of a mechanism of universal quantification over types (second-order) and is itself extended in CIC.} supports a mechanism of universal quantification over types which results in variables ranging over types, and binders for them. It yields that generic data types such as list, trees, etc. can be encoded. It follows that system $F$ is more expressive than STT. CIC adds universes to the system $F$, which leads to an improvement of its consistency strength. Adding dependent types in CIC enhances the computational power but does not affect its consistency strength. As a result, the expressive power of CIC is higher than STT.  
\par In the second step we have to prove that Church Simple Type Theory is at least as expressive as protothetic. In protothetic quantification is allowed on propositional variables and variables of propositional functors to any degree. It follows that protothetic is equivalent in expressive potential to a theory of propositional types \cite{Henkin63,Simon98}. Now, if we consider LO which includes protothetic, all symbols of LO can be substituted with variables that can be quantified, then LO is equivalent in expressive power to STT. It yields that protothetic which is a sub-theory of LO is at most as expressive as STT. Combining the two claims, we derive that CIC is at least as expressive as protothetic. 
\end{proof}
\par It turns out that Le{\'s}niewski's mereology can be expressed in CIC and thus, encoded in Coq. CIC \cite{Coq88,Coq90} is an intensional type theory using a universe for logic ($Prop$) and a hierarchy of universes for types whose low-level universe is $Type$ reflecting merely the Le{\'s}niewski's categorization with the respective categories of LO, $S$ (propositions) and $N$ (names). The Calculus of constructions (and therefore Coq) also includes an inference rule for equality between types called conversion rule. However, it does not include primarily an extensional equivalence and requires a way of defining an extensional conversion within an intensional system \cite{Seldin01,Oury05}. For that purpose setoids provide an efficient way to introduce extensionality which maintains the difference between identity and equivalence i.e., with an interpretation of intensional equality (the equality on the original set) and extensional equality (the equivalence relation on the new set). Setoids have been early introduced into the Coq theorem prover \cite{Hof95,Bart03}, and refined using type classes \cite{Sozeau09}. The encoding of setoids for extensionality in the paper will refer to the latter. The basic idea is to express mereology in Coq with some constraints: (i) mapping the category $S$ onto $Prop$ and the category of names to the type $N$, (ii) using dedicated equivalence relations for expressing functors and (iii) expressing extensionality through setoids.
\subsection{A Short Introduction to the Coq Theorem Prover}\label{coq}
The Coq language is a tool for developing mathematical specifications and proofs. As a specification language, Coq is both a higher order logic (quantifiers may be applied on natural numbers, on functions of arbitrary types, on propositions, predicates, types, etc.) and a typed lambda-calculus enriched with an extension of primitive recursion (further details are given in \cite{Bertot04}). The underlying logic of the Coq system is an intuitionist logic. This means that the proposition $A \: \vee \neg A$ is not taken for granted and, if it is needed, the user has to assume it explicitly. This allows to clarify the distinction between classical and constructive proofs. 
\par Its building blocks are terms and the basic relation is the typing relation. Coq is designed such that it ensures decidability of type checking. All logical judgments are typing judgments. The type-checker checks the correctness of proofs, that is, it checks that a data structure complies to its specification. The language of the Coq theorem prover consists in a sequence of declarations and definitions. A declaration associates a name with a qualification. Qualifications can be either logical propositions which reside in the universe\footnote{Also called \textit{sort}.} $Prop$ or abstract types which belong to the universe $Type$ (the universe $Type$ is stratified but this aspect is not relevant here).
\par Conversion rules such as $\beta$-reduction allow for term reductions which in that particular case, will formalize the substitution rule of protothetic. Standard equality in Coq is the Leibniz equality where propositionally equal terms are meant to be equivalent with respect to all their properties. Leibniz equality will be replaced by appropriate equivalence relations that are defined in setoids. The proof engine also provides an interactive proof assistant to build proofs using specific programs called tactics. Tactics are the cornerstone of proof-search in the process of theorem proving.
\subsection{Expressing Mereology in Coq}
Models of mereology that have been constructed so far focus on the interpretation of part-whole relations letting aside the interpretation of the copula $\varepsilon$ \cite{Clay68}. We depart from this assumption and encode the Le{\'s}niewski's systems in Coq based on a single axiom for the copula. Two semantic categories are defined, $Prop$ for propositional variables and $N$ for name variables. Any constant of any semantic category is well typed by means of formerly introduced symbols. The principle of extensionality for propositional types is defined as the morphism ($\mathsf{iff}$ refer here to the usual first-order condition):
\begin{coqdoccode}
\begin{small}
\coqdocemptyline
\coqdockw{Notation} \textgravedbl a $\equiv \equiv$ b\textacutedbl := (\coqdocdefinition{iff} \coqdocvar{a} \coqdocvariable{b})  (\coqdoctac{at} \coqdockw{level} \coqdocprojection{70)}.\coqdoceol
\coqdocemptyline
\end{small}
\end{coqdoccode}
The copula $\varepsilon$ is introduced by a single axiom.
\begin{coqdoccode}
\begin{small}
\coqdocemptyline
\coqdockw{Class} \coqdocrecord{N}.\coqdoceol
\coqdockw{Parameter} \coqdef{mereo setoidsv8.Mereology.N}{N}{\coqdocaxiom{epsilon}} : \coqdocvar{N} $\rightarrow$ \coqdocvar{N} $\rightarrow$ \coqdockw{Prop}.\coqdoceol
\coqdocemptyline
\coqdockw{Notation} \textgravedbl\coqdocvar{A } '$\varepsilon$'\coqdocvariable{ b}\textacutedbl := (\coqdocaxiom{epsilon} \coqdocvar{A} \coqdocprojection{b})  (\coqdocprojection{at} \coqdocprojection{level} \coqdocprojection{70}).\coqdoceol
\coqdocemptyline
\coqdockw{Axiom} \coqdef{mereo setoidsv8.Mereology.isEpsilon}{isEpsilon}{\coqdocaxiom{isEpsilon}} : \ensuremath{\forall} \coqdocvar{A} \coqdocvar{a},
  \coqdocvar{A} $\varepsilon$ \coqdocvar{a} $\equiv \equiv$ (({\coqdocnotation{\ensuremath{\exists}}} \coqdocvar{B}, \coqdocvariable{B} $\varepsilon$ \coqdocvariable{A}) \ensuremath{\land} 
(\ensuremath{\forall} \coqdocvar{C} \coqdocvar{D}, (\coqdocvar{C} $\varepsilon$ \coqdocvar{A} \ensuremath{\land} \coqdocvar{D} $\varepsilon$ \coqdocvar{A} \ensuremath{\rightarrow} \coqdocvariable{C} $\varepsilon$ \coqdocvariable{D})) \ensuremath{\land}  \coqdoceol
\coqdocindent{6.25cm}
(\coqdockw{\ensuremath{\forall}} \coqdocvar{C},  \coqdocvar{C} $\varepsilon$ \coqdocvar{A}  \ensuremath{\rightarrow} \coqdocvar{C} $\varepsilon$ \coqdocvar{a})). \coqdoceol
\coqdocemptyline
\end{small}
\end{coqdoccode}
The axiom includes three sub-axioms. The first one requires the existence of $A$, the second assumes the uniqueness of $A$ while the last one expresses the convergence of all which can be $A$. From the full ontology, we retain two kinds of equality between names:
\begin{coqdoccode}
\begin{small}
\coqdocemptyline
\coqdockw{Parameter} \coqdef{mereo core2.Mereology.singular equality}{singular\_equality}{\coqdocaxiom{singular\_equality}}   : \coqref{mereo core2.Mereology.N}{\coqdocvar{N}} \ensuremath{\rightarrow} \coqdocvar{N} \ensuremath{\rightarrow} \coqdockw{Prop}.\coqdoceol
\coqdockw{Parameter} \coqdef{mereo core2.Mereology.D4}{D4}{\coqdocaxiom{D4}}                  : \coqdockw{\ensuremath{\forall}} \coqdocvar{A} \coqdocvar{B},  \coqref{mereo core2.Mereology.singular equality}{\coqdocaxiom{singular\_equality}} \coqdocvariable{A} \coqdocvariable{B} $\equiv\equiv$ (\coqdocvariable{A} $\varepsilon$ \coqdocvariable{B} \ensuremath{\land} \coqdocvariable{B} $\varepsilon$ \coqdocvariable{A}).  \coqdocemptyline
\coqdockw{Parameter} \coqdef{mereo core2.Mereology.weak equality}{weak\_equality}{\coqdocaxiom{weak\_equality}}       : \coqref{mereo core2.Mereology.N}{\coqdocvar{N}} \ensuremath{\rightarrow} \coqref{mereo core2.Mereology.N}{\coqdocvar{N}} \ensuremath{\rightarrow} \coqdockw{Prop}.\coqdoceol
\coqdockw{Parameter} \coqdef{mereo core2.Mereology.D6}{D6}{\coqdocaxiom{D6}}                  : \coqdockw{\ensuremath{\forall}} \coqdocvar{a} \coqdocvar{b},  \coqref{mereo core2.Mereology.weak equality}{\coqdocaxiom{weak\_equality}} \coqdocvariable{a} \coqdocvariable{b} $\equiv\equiv$ (\coqdockw{\ensuremath{\forall}} \coqdocvar{A}, (\coqdocvariable{A} $\varepsilon$ \coqdocvariable{a}) $\equiv\equiv$ (\coqdocvariable{A} $\varepsilon$ \coqdocvariable{b})).
\coqdocemptyline
\end{small}
\end{coqdoccode}
Whereas D4 defines the identity of singulars $A$ and $B$ in terms of $\varepsilon$, D6 formalizes the weak equality between $a$ and $b$.
In LO, the law of extensionality for identities written $\forall A \: B \: \Phi, \; A = B \wedge  \Phi(A) \rightarrow \Phi(B)$ which requires the equality between singular names, states the generalization of the singular equality with the higher-order functor $\Phi$ \cite{Sobo60}. Such an higher-order extensionality can be captured in Coq on the basis of name equivalence relations. Unlike the singular equality (which is irreflexive), the functor for weak equality has the properties of an equivalence relation. It is easily derivable from definition $D6$, the \textsf{iff} morphism and basic tactics.
\begin{coqdoccode}
\begin{small}
\coqdocemptyline
\coqdockw{Lemma} \coqdef{mereo setoidsv8.Mereology.OntoT23}{OntoT23}{\coqdoclemma{OntoT23}} : \coqdocdefinition{reflexive} \coqdocvar{\_} \coqref{mereo setoidsv8.Mereology.weak equality}{\coqdocdefinition{weak\_equality}}.\coqdoceol
\coqdockw{Lemma} \coqdef{mereo setoidsv8.Mereology.OntoT24}{OntoT24}{\coqdoclemma{OntoT24}} : \coqdocdefinition{symmetric} \coqdocvar{\_} \coqref{mereo setoidsv8.Mereology.weak equality}{\coqdocdefinition{weak\_equality}}.\coqdoceol
\coqdockw{Lemma} \coqdef{mereo setoidsv8.Mereology.OntoT25}{OntoT25}{\coqdoclemma{OntoT25}} : \coqdocdefinition{transitive} \coqdocvar{\_} \coqref{mereo setoidsv8.Mereology.weak equality}{\coqdocdefinition{weak\_equality}}.\coqdoceol
\coqdocemptyline
\end{small}
\end{coqdoccode} 
It follows that the weak equality is an equivalence relation. Using setoids, we can state the extensionality lemma \begin{coqdoccode}\begin{small}\coqdoclemma{MereoT16}\end{small}\end{coqdoccode} using higher-order categories including names:
\begin{coqdoccode}
\begin{small}
\coqdocemptyline
\coqdocnoindent
\coqdockw{Definition} \coqdef{mereo core3.Mereology.Phi}{Phi}{\coqdocdefinition{Phi}} (\coqdocvar{E1}: \coqdocvar{N}\ensuremath{\rightarrow}\coqdocvar{N}\ensuremath{\rightarrow}\coqdockw{Prop})(\coqdocvar{E2}:\coqdocvar{N})(\coqdocvar{E3}:\coqdocvar{N}\ensuremath{\rightarrow}\coqdocvar{N})(\coqdocvar{E4}:\coqdocvar{N}) := \coqdocvariable{E1} \coqdocvariable{E2} (\coqdocvariable{E3} \coqdocvariable{E4}).\coqdoceol
\coqdocemptyline 
\coqdocnoindent
\coqdockw{Lemma} \coqdef{mereocore3.Mereology.MereoT16}{MereoT16}{\coqdoclemma{MereoT16}} : \coqdockw{\ensuremath{\forall}}(\coqdocvar{A}:\coqdocvar{N})(\coqdocvar{B}:\coqdocvar{N})(\coqdocvar{C}:\coqdocvar{N})(\coqdocvar{Phi}:\coqdocvar{N}\ensuremath{\rightarrow}\coqdocvar{N}), \coqdocvar{A}$\varepsilon$\coqdocvar{Phi} \coqdocvar{B} \ensuremath{\land} \coqdocaxiom{singular\_equality}\coqdocvar{B}\coqdocvar{C}
\ensuremath{\rightarrow} \coqdocvariable{A}$\varepsilon$\coqdocvariable{Phi} \coqdocvariable{C}.\coqdoceol
\coqdocemptyline
\end{small}
\end{coqdoccode}
The ontology described so far consists in an axiom and a collection of lemmas currently introduced. The syntactic domain of the ontology is the set of well-formed expressions which can be expressed  w.r.t. all semantic categories that have been defined. The semantic domain include a collection of objects together with a collection of names, each collection being possibly empty. 
\par The mereology adds a number of functors such as the \textsf{Kl} functor which describes the collective classes. For that purpose, two name forming functors are required, i.e., \textsf{pt} (for \textgravedbl part of\textacutedbl) and \textsf{el} (for \textgravedbl element of\textacutedbl) which belong to the semantic category $N \rightarrow N$.
\begin{coqdoccode}
\begin{small}
\coqdocemptyline
\coqdockw{Parameter} \coqdocaxiom{pt} \coqdocindent{1em}: \coqdocvar{N} \ensuremath{\rightarrow} \coqdocvar{N}.\coqdoceol
\coqdockw{Parameter} \coqdocaxiom{el} \coqdocindent{1.1em}: \coqdocvar{N} \ensuremath{\rightarrow} \coqdocvar{N}.\coqdoceol
\coqdocemptyline
\end{small}
\end{coqdoccode}
With such declarations, a partitive relation \begin{coqdoccode}\begin{small} \coqdocvar{A} $\varepsilon$ (\coqdocaxiom{pt} \coqdocvar{b})\end{small}\end{coqdoccode} (read \textgravedbl A is a part of b\textacutedbl) and a membership relation \begin{coqdoccode}\begin{small} \coqdocvar{A} $\varepsilon$ (\coqdocaxiom{el} \coqdocvar{b})\end{small}\end{coqdoccode} (read \textgravedbl A is an element of b\textacutedbl) are introduced. The former is taken as a primitive of mereology and two axioms govern its properties. They respectively state the asymmetry and transitivity of the \textsf{$\varepsilon$ pt} relation. Notice that \textsf{$\varepsilon$ pt}, \textsf{$\varepsilon$ el} and \textsf{$\varepsilon$ Kl} respectively correspond to what are termed \textgravedbl proper-part-of\textacutedbl, \textgravedbl part-of\textacutedbl and \textgravedbl sum\textacutedbl in Tarski's mereogeometry. If we were eliminating $\varepsilon$ we would remove the notion of distributive class \cite{Clay74} and consequently the interplay between the two notions of class, aspects of mereology which Le{\'s}niewski thought were so important that he constructed LO in order to describe them clearly. For the sake of simplicity, definition \begin{coqdoccode}\begin{small}\coqdocdefinition{ProperPart\_of}\end{small}\end{coqdoccode} represents an alias for handling \textit{proper part of} relations and should not be confused with a mereological definition using equivalence ($\equiv\equiv$). 
\begin{coqdoccode}
\begin{small}
\coqdocemptyline
\coqdockw{Definition} \coqdocdefinition{ProperPart\_of}   \coqdocindent{2.1em}:= \coqdockw{fun} \coqdocvar{A} \coqdocvar{B} \ensuremath{\Rightarrow} \coqdocvariable{A} $\varepsilon$ \coqdocaxiom{pt} \coqdocvariable{B}. \coqdoceol
\coqdocemptyline
\coqdockw{Axiom} \coqdocaxiom{asymmetric\_ProperPart}  \coqdocindent{0.9em}: \coqdockw{\ensuremath{\forall}} \coqdocvar{A} \coqdocvar{B}, \coqdocvariable{A} $\varepsilon$ \coqdocaxiom{pt} \coqdocvariable{B} \ensuremath{\rightarrow} \coqdocvariable{B} $\varepsilon$ (\coqdocaxiom{distinct} (\coqdocaxiom{pt} \coqdocvariable{A})). (*A1*)\coqdoceol
\coqdockw{Axiom} \coqdocaxiom{transitive\_PoperPart} \coqdocindent{1.8em}: \coqdocclass{Transitive} \coqdocdefinition{ProperPart\_of}. \coqdocindent{2.5cm} (*A2*)\coqdocemptyline
\end{small}
\end{coqdoccode}
The term \textsf{distinct} whose type is $N \rightarrow N$ is such that from $b$, the type \textsf{distinct} $b$ will capture the meaning of \textgravedbl to be different from b\textacutedbl. The next mereological definition follows the structure of ontological definitions and expresses \textgravedbl being an element of\textacutedbl as a name-forming functor.
\begin{coqdoccode}
\begin{small}
\coqdocemptyline
\coqdockw{Parameter} \coqdocaxiom{MD1}                : \coqdockw{\ensuremath{\forall}} \coqdocvar{A} \coqdocvar{B}, \coqdocvariable{A} $\varepsilon$ \coqdocaxiom{el} \coqdocvariable{B} $\equiv\equiv$ (\coqdocvariable{A} $\varepsilon$ \coqdocvariable{A} \ensuremath{\land} (\coqdocaxiom{singular\_equality} \coqdocvariable{A} \coqdocvariable{B} \ensuremath{\lor} 
\coqdocdefinition{ProperPart\_of} \coqdocvariable{A} \coqdocvariable{B})).\coqdoceol
\coqdocemptyline
\end{small}
\end{coqdoccode}
To formulate the next axioms, we first introduce the mereological class, which is also defined as a name forming functor in the usual way:
\begin{coqdoccode}
\begin{small}
\coqdocemptyline
\coqdockw{Parameter} \coqdocaxiom{MD2}                : \coqdockw{\ensuremath{\forall}} \coqdocvar{A} \coqdocvar{a}, \coqdocvariable{A} $\varepsilon$ \coqdocaxiom{Kl} \coqdocvariable{a} $\equiv\equiv$ 
(\coqdocvariable{A} $\varepsilon$ \coqdocvariable{A} \ensuremath{\land}
 (\ensuremath{\exists} \coqdocvar{B}, \coqdocvariable{B} $\varepsilon$ \coqdocvariable{a}) \ensuremath{\land}
(\coqdockw{\ensuremath{\forall}} \coqdocvar{B}, \coqdocvariable{B} $\varepsilon$ \coqdocvariable{a} \ensuremath{\rightarrow} \coqdocvariable{B} $\varepsilon$ \coqdocaxiom{el} \coqdocvariable{A}) \ensuremath{\land}  \coqdoceol
\coqdocindent{6.75cm}
(\coqdockw{\ensuremath{\forall}} \coqdocvar{B}, \coqdocvariable{B} $\varepsilon$ \coqdocaxiom{el} \coqdocvariable{A} \ensuremath{\rightarrow} \ensuremath{\exists} \coqdocvar{C} \coqdocvar{D}, (\coqdocvariable{C} $\varepsilon$ \coqdocvariable{a} \ensuremath{\land} \coqdocvariable{D} $\varepsilon$ \coqdocaxiom{el} \coqdocvariable{C} \ensuremath{\land} \coqdocvariable{D} $\varepsilon$ \coqdocaxiom{el} \coqdocvariable{B}))).\coqdoceol
\coqdocemptyline
\end{small}
\end{coqdoccode}
The third axiom states that a mereological class is unique.
\begin{coqdoccode}
\begin{small}
\coqdocemptyline
\coqdockw{Axiom} \coqdocaxiom{Kl\_uniqueness}    : \coqdockw{\ensuremath{\forall}} \coqdocvar{A} \coqdocvar{B} \coqdocvar{a}, (\coqdocvariable{A} $\varepsilon$ \coqdocaxiom{Kl} \coqdocvariable{a} \coqdocnotation{\ensuremath{\land}} \coqdocvariable{B} $\varepsilon$ \coqdocaxiom{Kl} \coqdocvariable{a}) \ensuremath{\rightarrow} 
\coqdocdefinition{singular\_equality} \coqdocvariable{A} \coqdocvariable{B}. (*A3*)
\coqdocemptyline
\end{small}
\end{coqdoccode}
The fourth one postulates the class existence.
\begin{coqdoccode}
\begin{small}
\coqdocemptyline
\coqdockw{Axiom} \coqdocaxiom{Kl\_existence}     : \coqdockw{\ensuremath{\forall}} \coqdocvar{A} \coqdocvar{a}, \coqdocvariable{A} $\varepsilon$ \coqdocvariable{a} \ensuremath{\rightarrow} \coqdocnotation{\ensuremath{\exists}} \coqdocvar{B}, \coqdocvariable{B} $\varepsilon$ \coqdocaxiom{Kl} \coqdocvariable{a}. \coqdocindent{3.2cm}(*A4*) 
\coqdocemptyline
\end{small}
\end{coqdoccode}
No additional rules of inference are added in mereology. Many lemmas can be proved on the basis of mereological axioms. We only mention the following lemmas since they illustrate the crucial difference between set theory and Le{\'s}niewski's mereology.
\begin{coqdoccode}
\begin{small}
\coqdocemptyline
\coqdockw{Lemma} \coqdef{mereo core3.Mereology.MereoT26}{MereoT26}{\coqdoclemma{MereoT26}} : \coqdockw{\ensuremath{\forall}} \coqdocvar{A} \coqdocvar{a}, \coqdocvariable{A} $\varepsilon$ \coqdocaxiom{Kl} \coqdocvar{a} \ensuremath{\rightarrow} \coqdocvar{A} $\varepsilon$ \coqdocaxiom{Kl} (\coqdocaxiom{Kl} \coqdocvar{a}). \coqdoceol
\coqdockw{Lemma} \coqdef{mereo core3.Mereology.MereoT27}{MereoT27}{\coqdoclemma{MereoT27}} : \ensuremath{\forall} \coqdocvar{A} \coqdocvar{a}, \coqdocvariable{A} $\varepsilon$ \coqdocaxiom{Kl} (\coqdocaxiom{Kl} \coqdocvar{a}) \ensuremath{\rightarrow} \coqdocvar{A} $\varepsilon$ \coqdocaxiom{Kl} \coqdocvar{a}. \coqdoceol
\coqdockw{Lemma} \coqdef{mereo setoidsv12.Mereology.MereoT29}{MereoT29}{\coqdoclemma{MereoT29}} : \coqdockw{\ensuremath{\forall}} \coqdocvar{A}, \coqdocnotation{\ensuremath{\lnot}}(\coqdocvariable{A} $\varepsilon$ \coqref{mereo setoidsv12.Mereology.Kl}{\coqdocaxiom{Kl}} \coqref{mereo setoidsv12.Mereology.empty}{\coqdocaxiom{empty}}).\coqdoceol
\coqdocemptyline
\end{small}
\end{coqdoccode}
Lemmas \begin{coqdoccode}\begin{small}\coqdoclemma{MereoT26}\end{small}\end{coqdoccode} and \begin{coqdoccode}\begin{small}\coqdoclemma{MereoT27}\end{small}\end{coqdoccode} show that in mereology, the class of $a$ and the class of the class of $a$ stand for the same object, which is in contrast with set theory. Lemma \begin{coqdoccode}\begin{small}\coqdoclemma{MereoT29}\end{small}\end{coqdoccode} asserts that the class resulting from the empty name does not exist. The system of mereology is consistent which can be proved by many ways (e.g., with an appropriate model for mereology such as \cite{Clay68} or from protothetic \cite{Lejewski69}). Additional name-forming functors can be added to the theory to enhance its expressiveness. Le{\'s}niewski added the \begin{coqdoccode}\begin{small}\coqdocaxiom{coll}\end{small}\end{coqdoccode}(collection of) functors, while Lejewski added the \begin{coqdoccode}\begin{small}\coqdocaxiom{ov}\end{small}\end{coqdoccode}
(overlap) functor.
\begin{coqdoccode}
\begin{small}
\coqdocemptyline
\coqdockw{Parameter} \coqdef{mereo core3.Mereology.coll}{coll}{\coqdocaxiom{coll}} \coqdocindent{1.9em}: \coqdocvar{N} \ensuremath{\rightarrow} \coqdocvar{N}.  \coqdocindent{3.05cm}(* collection of *)\coqdoceol
\coqdockw{Parameter} \coqdocaxiom{ov}      \coqdocindent{2.3em}: \coqdocvar{N} \ensuremath{\rightarrow} \coqdocvar{N}. \coqdocindent{3.05cm}(* overlap *)\coqdoceol  
\coqdocemptyline
\coqdockw{Parameter} \coqdocaxiom{MD3}         : \coqdockw{\ensuremath{\forall}} \coqdocvar{P} \coqdocvar{a}, \coqdocvariable{P} $\varepsilon$ \coqdocaxiom{coll} \coqdocvariable{a} $\equiv\equiv$ (\coqdocvariable{P} $\varepsilon$ \coqdocvariable{P} \ensuremath{\land} \coqdockw{\ensuremath{\forall}} \coqdocvar{Q}, \coqdocvariable{Q} $\varepsilon$ \coqdocaxiom{el} \coqdocvar{P} \ensuremath{\rightarrow} \ensuremath{\exists} \coqdocvar{C} \coqdocvar{D}, (\coqdocvar{C} $\varepsilon$ \coqdocvar{a} \ensuremath{\land} \coqdocvar{C} $\varepsilon$ \coqdocaxiom{el} \coqdocvar{P} \ensuremath{\land} 
\coqdoceol
\coqdocindent{7cm} 
 \coqdocvar{D} $\varepsilon$ \coqdocaxiom{el} \coqdocvar{C} \ensuremath{\land} \coqdocvar{D} $\varepsilon$ \coqdocaxiom{el} \coqdocvar{Q})).\coqdoceol
\coqdockw{Parameter} \coqdocaxiom{MD7}                : \coqdockw{\ensuremath{\forall}} \coqdocvar{P} \coqdocvar{Q}, \coqdocvariable{P} $\varepsilon$ \coqdocaxiom{ov} \coqdocvariable{Q} $\equiv\equiv$ (\coqdocvariable{P} $\varepsilon$ \coqdocvariable{P} \ensuremath{\land} \ensuremath{\exists} \coqdocvar{C}, (\coqdocvariable{C} $\varepsilon$ \coqdocaxiom{el} \coqdocvariable{P} \ensuremath{\land} \coqdocvariable{C} $\varepsilon$ \coqdocaxiom{el} \coqdocvariable{Q})).\coqdoceol
\coqdocemptyline
\end{small}
\end{coqdoccode} 
A useful lemma in the following states that if an object $P$ is an $a$, then $P$ is a collection of objects $a$.
\begin{coqdoccode}
\begin{small}
\coqdocemptyline
\coqdockw{Lemma} \coqdocaxiom{XIII} 		: \ensuremath{\forall} \coqdocvar{a} \coqdocvar{P}, \coqdocvar{P} $\varepsilon$ \coqdocvar{a} $\rightarrow$ (\coqdocvar{P} $\varepsilon$ \coqdocaxiom{coll} \coqdocvar{a}).
\coqdocemptyline
\end{small}
\end{coqdoccode} 
In addition, it is provable that if an object is a sub-collection of another, then it is equivalent to say that it is a part of this object.
\begin{coqdoccode}
\begin{small}
\coqdocemptyline
\coqdockw{Lemma} \coqdocaxiom{MereoT44} 		: \ensuremath{\forall} \coqdocvar{A} \coqdocvar{B}, \coqdocvar{B} $\varepsilon$  \coqdocaxiom{el} \coqdocvar{A} $\equiv\equiv$ (\coqdocvar{B} $\varepsilon$ \coqdocaxiom{subcoll} \coqdocvar{A}).
\coqdocemptyline
\end{small}
\end{coqdoccode} 
\par To clarify the notion of class we describe an example extracted from \cite{Gessler05} in figure \ref{fig1} which represents a rectangle labeled $R$ including geometric parts. Let us suppose that the generic name $a$ describes the squares of $R$. In other words, we get:
$A \: \varepsilon \: Kl (square \: of \: R)$ with the obvious meaning \textgravedbl A is the class of squares of R\textacutedbl.
\begin{figure}[ht]
\begin{center}
\includegraphics[scale=0.5]{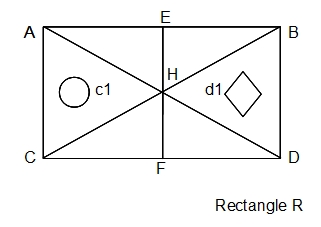}
\caption{Example of a collective class.}\label{fig1}
\end{center}
\end{figure}
The four parts of the class definition \texttt{MD2} are easily provable:

\begin{tabular}{l p{11.1cm}}
1. & the name $A$ denotes an object, i.e., rectangle $R$,\\
2. & there exists a $B$ such that $B \: \varepsilon \: (square \: of \: R)$ (e.g., $AECF$),\\
3. & any square of $R$ is an element of $A$ e.g., $EBFD \: \varepsilon \: el \: A$,\\
4. & for any element of $R$, e.g., the triangle $BHD$, an object $C$ which is a square of $R$ should exists, while another object $D$ should also exist as both an element of $C$ and an element of $B$. If $C$ denotes $EBFD$ and $D$, the diamond $d1$, then we easily show that $d1$ is both element of $EBFD$ and $BHD$. We can also consider a more complex object. Let us define $B$ as the collective object including $c1$ (a circle) and $BHD$. Then, there exists $C$ which is a square of $R$ e.g., $EBFD$ and $D$ which is both element of $C$ and $B$, e.g., the diamond $d1$.\\
\end{tabular}

\vspace{1em}
\begin{figure}[ht]
\begin{center}
\includegraphics[scale=0.35]{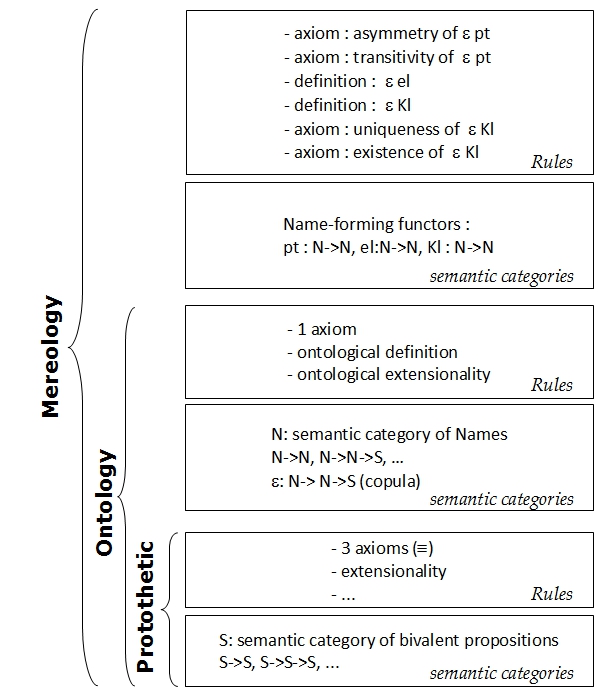}
\caption{An overview of Le{\'s}niewski's Mereology.}\label{fig2}
\end{center}
\end{figure}
\par We can observe that the content of collective classes is much more expressive than distributive classes of set theory. The syntactical aspect of Le{\'s}niewski's Mereology is summarized in figure \ref{fig2}. Each sub-theory defines (i) available categories and well-formed expressions, (ii) the assumed set of axioms and (ii) the set of inference rules. 
\section{The Type-theoretical Mereogeometry}\label{Mereogeometry}
We first recall the set of definitions suggested by Tarski to support the mereogeometry of solids. In the second subsection, we discuss some issues about Tarski's work while the last subsection presents the type-theoretical account of Tarski mereogeometry.
\subsection{Tarski's Definitions}
Tarski starts out from some mereological notions such as being a part of and sum and then, provides a minimal set of axioms constituting the mereological part on which the rest of the theory relies. This part will be replaced in the present work by the mereology developed in section \ref{TypeMereo}. Apart from this mereological part, Tarski takes \textit{sphere} (which will be called ball here) as the only primitive notion specific to the geometry of solids. Assuming these notions suffices to formulate a set of nine geometrical definitions which constitutes the basis of mereogeometry.
\begin{defi}
The ball $A$ is externally tangent to the ball $B$ if (i) the ball $A$ is disjoint from the ball $B$ and (ii) given two balls $X$ and $Y$ containing as part the ball $A$ and disjoint from the ball $B$, at least one of them is part of the other. 
\end{defi}
\begin{defi}
The ball $A$ is internally tangent to the ball $B$ if (i) the ball $A$ is a proper part of the ball $B$ and (ii) given two balls $X$ and $Y$ containing the ball $A$ as a part and forming part of the ball $B$, at least one of them is a part of the other.
\end{defi}
\begin{defi}
The balls $A$ and $B$ are externally diametrically tangent to the ball $C$ if (i) each of the balls $A$ and $B$ is externally tangent to the ball $C$ and (ii) given two balls $X$ and $Y$ disjoint from the ball $C$ and such that $A$ is part of $X$ and $B$ a part of $Y$, $X$ is disjoint from $Y$.
\end{defi}
\begin{defi}
The balls $A$ and $B$ are internally diametrically tangent to the ball $C$ if (i) each of the balls $A$ and $B$ is internally tangent to the ball $C$ and (ii) given two balls $X$ and $Y$ disjoint from the ball $C$ and such that the ball $A$ is externally tangent to $X$ and $B$ to $Y$, $X$ is disjoint from $Y$.
\end{defi}
\begin{defi}
The ball $A$ is concentric with the ball $B$ if one of the following conditions is satisfied: (i) the balls $A$ and $B$ are identical, (ii) the ball $A$ is a proper part of $B$ and besides, given two balls $X$ and $Y$ externally diametrically tangent to $A$ and internally tangent to $B$, these balls are internally diametrically tangent to $B$ and (iii) the ball $B$ is a proper part of $A$ and besides, given two balls $X$ and $Y$ externally diametrically tangent to $B$ and internally tangent to $A$, these balls are internally diametrically tangent to $A$.
\end{defi}
\begin{defi}
A point is the class of all balls which are concentric with a given ball.
\end{defi}
\begin{defi}
The points $A$ and $B$ are equidistant from the point $C$ if there exists a ball $X$ which belongs as element to the point $C$ and which satisfies the following condition: no ball $Y$ belonging as element to the point $A$ or to the point $B$ is a part of $X$ or is disjoint from $X$.
\end{defi}
\begin{defi}
A solid is an arbitrary sum of balls.
\end{defi}
\begin{defi}
The point $P$ is an interior point of the solid $B$ if there exists a ball $A$ which is at the same time an element of the point $P$ and a part of the solid $B$. 
\end{defi}
The 2D version of these definitions is illustrated in figure \ref{fig3}.
\begin{figure}[ht]
\begin{center}
\includegraphics[scale=0.5]{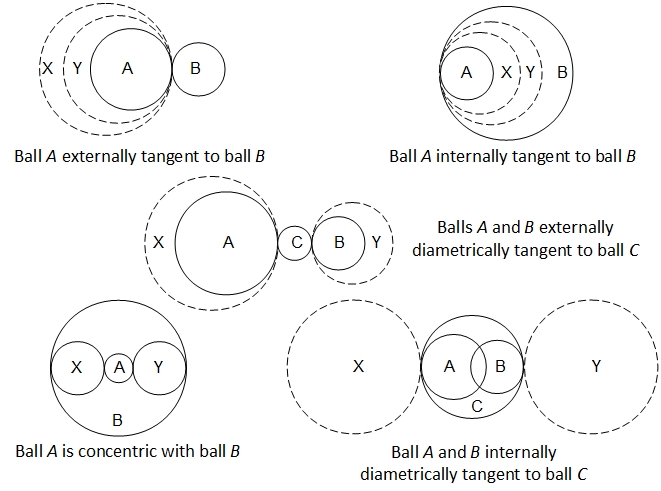}
\caption{Tarski's primitive definitions.}\label{fig3}
\end{center}
\end{figure}
\subsection{Some issues in Tarski's Work}
While in the sketchy paper of Tarski \cite{Tarski29} the logical background seems rather unclear \cite{Grusz08,Betti12} we depart from all approaches stemming from a set-based framework and rather argue that an axiomatization of Tarski's work based solely on Le{\'s}niewski's Mereology is possible. Instead of introducing axioms for part-whole relations, we rather rely on the Le{\'s}niewski's mereology together with its set of axioms and theorems. In fact Tarski applies Le{\'s}niewski's mereology and does not extend it. We will justify these claims by discussing the successive issues highlighted in the recent survey about Tarski's work \cite{Betti12} where the author points out that the role of Le{\'s}niewski's systems is marginal. 
\par As argued by the author, Tarski's work can be given a meaningful account to some otherwise logically unclear parts of it, only if one assume a type theoretical approach. We have shown in section \ref{TypeMereo} that using a type-theoretical background substituting protothetic in Le{\'s}niewski's mereology is a possible approach that is coherent with the work of Henkin \cite{Henkin63}.
Another remark concerns the notions of being a part of and sum. While being a part of is said to be a primitive notion in Lesniewski's Mereology, but sum is not. We advocate that the references of Tarski to Le{\'s}niewski's work are syntactically exact (e.g., the word \textgravedbl sum\textacutedbl in definition 8 refers to the definition of \textgravedbl collection\textacutedbl in mereology \cite{Sinisi83}). 
\par Then, the author explains that the notion of domain of discourse remains unclear, i.e., that the range of the quantifiers can be that of solids or that of balls. We rather advocate for a more uniform solution in which balls and solids are names, and more precisely plurals in the sense of Le{\'s}niewski. Such an assumption yields a quantification over variables that are always of type $N$ (name) and avoids the issues concerning the domain of discourse. She also says that Tarski is freely mixing notions from Le{\'s}niewski's mereology (e.g., collective classes) with those of the system of \textit{Principia mathematica}, e.g., distributive classes. However, in Le{\'s}niewski's systems collective and distributive classes coexist \cite{Clay74}. It follows that instead of expressing Tarski's mereogeometry in a set-based framework \cite{Grusz08}, a more coherent picture can be proposed using Le{\'s}niewski's axiomatization. Collective classes are addressed with the \begin{coqdoccode}\begin{small}$\varepsilon$ \coqdocvar{Kl}\end{small}\end{coqdoccode} functor while distributive classes require new categories and definitions (in the form of equivalences). The Russellian notion of class that appears in definition 6, is solved using a definition in the sense of Le{\'s}niewski by constraining the quantification through concentricity (see next subsection).
\par A first benefit of the proposed type-theoretical formulation is that some axioms of Mereogeometry can be proved as theorems. This fact enhances the soundness and simplifies the whole theory. A second advantage is that many critiques of Tarski's paper can be avoided such as the difficulty to assess the domain of discourse (see, e.g., \cite{Betti12}).
\subsection{Expressing Definitions of Mereogeometry in Coq}
The first commitment concerns balls\footnote{Called spheres in Tarski's paper.} and solids. These concepts are subject to discussion for they are seen as distinct distributive classes. Using the nominalist view of Le{\'s}niewski, these names are considered as constant plurals, since (i) they denote a constant plurality and (ii) they refer to collections of objects. In other words, they are instances of names ($N$) which are constrained to appear as the right argument of the copula ($\varepsilon$).
In LO, constant plurals (e.g., \textgravedbl empty name\textacutedbl or \textgravedbl universal name\textacutedbl) are defined in a similar way: their category must be defined first (i.e., $N$) and then, a definition explains their property in the form of an equivalence. Since \textgravedbl balls\textacutedbl is a primitive, no specific property is required. Alternatively, solids does not refer to a primitive name. As a consequence, it belongs to the category $N$ and it is defined as a particular collection of balls (see definition 8 below). 
\begin{coqdoccode}
\begin{small}
\coqdocemptyline
\coqdockw{Parameter} \coqdocaxiom{balls}  \coqdocindent{1.67cm}: \coqdocvar{N}. \coqdoceol
\coqdockw{Parameter} \coqdocaxiom{solids}  \coqdocindent{1.54cm}: \coqdocvar{N}. \coqdoceol
\coqdocemptyline
\end{small}
\end{coqdoccode}
In such a way the domain of discourse is that of truth values and names (basic categories) with a restriction to balls and solids. Notice that the restriction is not given by types but rather by theorems which directly constrain their use.
Using the \begin{coqdoccode}\begin{small}$\varepsilon$ \coqdocvar{pt}\end{small}\end{coqdoccode} functor, the four axioms of mereology and derived theorems, then it is easy to prove that the relation \begin{coqdoccode}\begin{small}$\varepsilon$\coqdocvar{el}\end{small}\end{coqdoccode} is a partial order:
\begin{coqdoccode}
\begin{small}
\coqdocemptyline
\coqdockw{Parameter} \coqdocaxiom{PartOf}  : \coqdocdefinition{relation} \coqdocvar{N}.\coqdoceol
\coqdockw{Parameter} \coqdocaxiom{MD11}               : \coqdockw{\ensuremath{\forall}} \coqdocvar{A} \coqdocvar{a}, \coqdocaxiom{PartOf} \coqdocvariable{A} \coqdocvariable{a} $\equiv\equiv$ (\coqdocvariable{A} $\varepsilon$ \coqdocaxiom{el} \coqdocvariable{a}).\coqdoceol
\coqdocemptyline
\coqdockw{Lemma} \coqdoclemma{Reflexive\_Element\_of} \coqdocindent{3.80em}: \coqdocclass{Reflexive} \coqdocaxiom{PartOf}. \coqdoceol
\coqdockw{Lemma} \coqdoclemma{AntiSymmetric\_Element\_of} \coqdocindent{1.30em}: \coqdocclass{Antisymmetric} \coqdocaxiom{PartOf}.\coqdoceol
\coqdockw{Lemma} \coqdoclemma{Transitive\_Element\_of} \coqdocindent{3.50em}: \coqdocclass{Transitive} \coqdocaxiom{PartOf}.\coqdoceol
\coqdocemptyline
\coqdockw{Theorem} \coqdoclemma{PartOf\_is\_partial\_order} \coqdocindent{2.15em}: \coqdocclass{POrder} \coqdocaxiom{PartOf}.\coqdoceol
\coqdocemptyline
\end{small}
\end{coqdoccode}
We first introduce the short-hand symbol $\leqslant$ which stands for \begin{coqdoccode}\begin{small}$\varepsilon$ \coqdocaxiom{el}\end{small}\end{coqdoccode}. Then, several relations among balls are defined such as concentricity, relying on the intended interpretation of the primitives. Points are defined as (mereological) collections of concentric balls. Equidistance among points makes use of properties of concentric balls while Euclidean axioms are able to constrain equidistance. We introduce successively:
\begin{coqdoccode}
\begin{small}
\coqdocemptyline
\coqdockw{Parameter} \coqdocaxiom{et}  \coqdocindent{1.54cm}: \coqdocvar{N} \ensuremath{\rightarrow} \coqdocvar{N}. \coqdoceol
\coqdockw{Parameter} \coqdocaxiom{it}  \coqdocindent{1.6cm}: \coqdocvar{N} \ensuremath{\rightarrow} \coqdocvar{N}.
\coqdoceol
\coqdockw{Parameter} \coqdocaxiom{edt}    \coqdocindent{1.36cm}: \coqdocvar{N} \ensuremath{\rightarrow} \coqdocvar{N} \ensuremath{\rightarrow} \coqdocvar{N}. \coqdoceol
\coqdockw{Parameter} \coqdef{mereogeometry.TarskiGeometry.idt}{idt}{\coqdocaxiom{idt}} \coqdocindent{1.43cm}: \coqdocvar{N} \ensuremath{\rightarrow} \coqdocvar{N} \ensuremath{\rightarrow} \coqdocvar{N}. \coqdoceol
\coqdockw{Parameter} \coqdocaxiom{con}   \coqdocindent{1.33cm}: \coqdocvar{N} \ensuremath{\rightarrow} \coqdocvar{N}. \coqdoceol
\coqdockw{Parameter} \coqdocaxiom{point} \coqdocindent{1.08cm}: \coqdocvar{N} \ensuremath{\rightarrow} \coqdocvar{N}. \coqdoceol
\coqdockw{Parameter} \coqdocaxiom{equid} \coqdocindent{1.06cm}: \coqdocvar{N} \ensuremath{\rightarrow} \coqdocvar{N} \ensuremath{\rightarrow} \coqdocvar{N}. \coqdoceol
\coqdockw{Parameter} \coqdocaxiom{ipoint} \coqdocindent{1.02cm}: \coqdocvar{N} \ensuremath{\rightarrow} \coqdocvar{N}. 
\coqdocemptyline
\end{small}
\end{coqdoccode}
\par The relations of external tangency (\begin{coqdoccode}\begin{small}$\varepsilon$ \coqdocaxiom{et}\end{small}\end{coqdoccode}), internal tangency (\begin{coqdoccode}\begin{small}$\varepsilon$ \coqdocaxiom{it}\end{small}\end{coqdoccode}), external diametricity (\begin{coqdoccode}\begin{small}$\varepsilon$ \coqdocaxiom{edt}\end{small}\end{coqdoccode}), internal diametricity (\begin{coqdoccode}\begin{small}$\varepsilon$ \coqdocaxiom{idt}\end{small}\end{coqdoccode}), concentricity (\begin{coqdoccode}\begin{small}$\varepsilon$ \coqdocaxiom{con}\end{small}\end{coqdoccode}), point (\begin{coqdoccode}\begin{small}\coqdocaxiom{point}\end{small}\end{coqdoccode}), equidistance (\begin{coqdoccode}\begin{small}$\varepsilon$ \coqdocaxiom{equi}\end{small}\end{coqdoccode}) and interior point (\begin{coqdoccode}\begin{small}$\varepsilon$ \coqdocaxiom{ipoint}\end{small}\end{coqdoccode}) are defined using the already defined name functors. 
\par Definition 1: external tangency.
\begin{coqdoccode}
\begin{small}
\coqdocemptyline
\coqdockw{Parameter} \coqdocaxiom{ET}      : \ensuremath{\forall} \coqdocvar{A} \coqdocvar{B}, \coqdocvar{A} $\varepsilon$ \coqdocaxiom{et} \coqdocvar{B}  $\equiv\equiv$ ((\coqdocvar{A} $\varepsilon$ \coqdocvar{balls}) \ensuremath{\land} (\coqdocvar{B} $\varepsilon$ \coqdocvar{balls}) \ensuremath{\land} (\coqdocvar{A} $\varepsilon$ \coqdocaxiom{ext} \coqdocvar{B})  \ensuremath{\land} \coqdoceol 
\coqdocindent{8.9em} \ensuremath{\forall} \coqdocvar{X} \coqdocvar{Y}, ((\coqdocvar{X} $\varepsilon$ \coqdocvar{balls}) 
\ensuremath{\land} (\coqdocvar{Y} $\varepsilon$ \coqdocvar{balls}) \ensuremath{\land} 
(\coqdocvar{A} $\leqslant$ \coqdocvar{X} \ensuremath{\land} \coqdocvar{X} $\varepsilon$ \coqdocaxiom{ext} \coqdocvar{B}) \ensuremath{\land}
\coqdoceol 
\coqdocindent{8.9em} (\coqdocvar{A} $\leqslant$ \coqdocvar{Y} \ensuremath{\land} \coqdocvar{Y} $\varepsilon$ \coqdocaxiom{ext} \coqdocvar{B})) $\rightarrow$ (\coqdocvar{X} $\leqslant$ \coqdocvar{Y} \ensuremath{\lor} \coqdocvar{Y} $\leqslant$ \coqdocvar{X})).\coqdoceol
\coqdocemptyline
\end{small}
\end{coqdoccode} 
Definition 2: internal tangency.
\begin{coqdoccode}
\begin{small}
\coqdocemptyline
\coqdockw{Parameter} \coqdocaxiom{IT}      : \ensuremath{\forall} \coqdocvar{A} \coqdocvar{B}, \coqdocvar{A} $\varepsilon$ \coqdocaxiom{it} \coqdocvar{B}  $\equiv\equiv$ ((\coqdocvar{A} $\varepsilon$ \coqdocvar{balls}) \ensuremath{\land} (\coqdocvar{B} $\varepsilon$ \coqdocvar{balls}) \ensuremath{\land} (\coqdocvar{A} $<$ \coqdocvar{B})  \ensuremath{\land} \coqdoceol 
\coqdocindent{8.9em} \ensuremath{\forall} \coqdocvar{X} \coqdocvar{Y}, ((\coqdocvar{X} $\varepsilon$ \coqdocvar{balls}) 
\ensuremath{\land} (\coqdocvar{Y} $\varepsilon$ \coqdocvar{balls}) \ensuremath{\land} 
(\coqdocvar{A} $\leqslant$ \coqdocvar{X} \ensuremath{\land} \coqdocvar{X} $\leqslant$ \coqdocvar{B}) \ensuremath{\land}
\coqdoceol 
\coqdocindent{8.9em} (\coqdocvar{A} $\leqslant$ \coqdocvar{Y} \ensuremath{\land} \coqdocvar{Y} $\leqslant$ \coqdocvar{B})) $\rightarrow$ (\coqdocvar{X} $\leqslant$ \coqdocvar{Y} \ensuremath{\lor} \coqdocvar{Y} $\leqslant$ \coqdocvar{X})).\coqdoceol
\coqdocemptyline
\end{small}
\end{coqdoccode} 
Definition 3: external diametrical tangency.
\begin{coqdoccode}
\begin{small}
\coqdocemptyline
\coqdockw{Parameter} \coqdocaxiom{EDT}  : \ensuremath{\forall} \coqdocvar{A} \coqdocvar{B} \coqdocvar{C}, \coqdocvar{A} $\varepsilon$ \coqdocaxiom{edt} \coqdocvar{B} \coqdocvar{C} $\equiv\equiv$ ((\coqdocvar{A} $\varepsilon$ \coqdocvar{balls}) \ensuremath{\land} (\coqdocvar{B} $\varepsilon$ \coqdocvar{balls}) \ensuremath{\land} \coqdoceol
\coqdocindent{8.9em}
(\coqdocvar{C} $\varepsilon$ \coqdocvar{balls}) \ensuremath{\land} (\coqdocvar{B} $\varepsilon$ \coqdocaxiom{et} \coqdocvar{A}) \ensuremath{\land} (\coqdocvar{C} $\varepsilon$ \coqdocaxiom{et} \coqdocvar{A}) \ensuremath{\land} 
\ensuremath{\forall} \coqdocvar{X} \coqdocvar{Y}, ((\coqdocvar{X} $\varepsilon$ \coqdocaxiom{balls}) \ensuremath{\land} 
\coqdoceol
\coqdocindent{8.90em} (\coqdocvar{Y} $\varepsilon$ \coqdocvar{balls} ) \ensuremath{\land} (\coqdocvar{B} $\leqslant$ \coqdocvar{X} \ensuremath{\land} \coqdocvar{X} $\varepsilon$ \coqdocaxiom{ext} \coqdocvar{A} ) \ensuremath{\land} (\coqdocvar{C} $\leqslant$ \coqdocvar{Y}  \ensuremath{\land} \coqdocvar{Y} $\varepsilon$ \coqdocaxiom{ext} \coqdocvar{A}))  \coqdoceol
\coqdocindent{8.90em} $\rightarrow$  (\coqdocvar{X} $\varepsilon$ \coqdocaxiom{ext} \coqdocvar{Y})).\coqdoceol
\coqdocemptyline
\end{small}
\end{coqdoccode} 
Definition 4: internal diametrical tangency.
\begin{coqdoccode}
\begin{small}
\coqdocemptyline
\coqdockw{Parameter} \coqdocaxiom{IDT}   : \ensuremath{\forall} \coqdocvar{A} \coqdocvar{B} \coqdocvar{C}, \coqdocvar{A} $\varepsilon$\coqdocaxiom{idt} \coqdocvar{B} \coqdocvar{C} $\equiv\equiv$ ((\coqdocvar{A} $\varepsilon$ \coqdocvar{balls}) \ensuremath{\land} (\coqdocvar{B} $\varepsilon$ \coqdocvar{balls}) \ensuremath{\land} (\coqdocvar{C} $\varepsilon$ \coqdocvar{balls})
\coqdoceol
\coqdocindent{8.90em} \ensuremath{\land} (\coqdocvar{B} $\varepsilon$ \coqdocaxiom{it} \coqdocvar{A}) \ensuremath{\land} (\coqdocvar{C} $\varepsilon$ \coqdocaxiom{it} \coqdocvar{A}) \ensuremath{\land} 
\ensuremath{\forall} \coqdocvar{X} \coqdocvar{Y}, (((( \coqdocvar{X} $\varepsilon$ \coqdocvar{balls}) \ensuremath{\land} (\coqdocvar{Y} $\varepsilon$ \coqdocvar{balls}) \coqdoceol
\coqdocindent{8.90em}
\ensuremath{\land} (\coqdocvar{X} $\varepsilon$ \coqdocaxiom{ext} \coqdocvar{A}) \ensuremath{\land} (\coqdocvar{Y} $\varepsilon$ \coqdocaxiom{ext} \coqdocvar{A}) \ensuremath{\land} (\coqdocvar{B} $\varepsilon$ \coqdocaxiom{ext} \coqdocvar{X}) \ensuremath{\land} (\coqdocvar{C} $\varepsilon$ \coqdocaxiom{ext} \coqdocvar{Y}))
\coqdoceol
\coqdocindent{8.90em}
 $\rightarrow$ (\coqdocvar{X} $\varepsilon$ \coqdocaxiom{ext} \coqdocvar{Y})).\coqdoceol
\coqdocemptyline
\end{small}
\end{coqdoccode} 
Definition 5: concentric balls.
\begin{coqdoccode}
\begin{small}
\coqdocemptyline
\coqdockw{Parameter} \coqdocaxiom{CON}      : \ensuremath{\forall} \coqdocvar{A} \coqdocvar{B}, \coqdocvar{A} $\varepsilon$ \coqdocaxiom{con} \coqdocvar{B} $\equiv\equiv$ ((\coqdocvar{A} $\varepsilon$ \coqdocvar{balls}) \ensuremath{\land} (\coqdocvar{B} $\varepsilon$ \coqdocvar{balls})) \ensuremath{\land} \coqdocaxiom{singular\_}
\coqdoceol
 \coqdocindent{8.90em}
\coqdocaxiom{equality} \coqdocvar{A} \coqdocvar{B} \ensuremath{\lor} 
(\coqdocvar{A} $<$ \coqdocvar{B} \ensuremath{\land} \ensuremath{\forall} \coqdocvar{X} \coqdocvar{Y}, ((\coqdocvar{X} $\varepsilon$ \coqdocvar{balls}) \ensuremath{\land} (\coqdocvar{Y} $\varepsilon$ \coqdocvar{balls}) \ensuremath{\land}
\coqdoceol
 \coqdocindent{8.90em}
(\coqdocvar{A} $\varepsilon$ \coqdocaxiom{edt} \coqdocvar{X} \coqdocvar{Y}) \ensuremath{\land} (\coqdocvar{X} $\varepsilon$ \coqdocaxiom{it} \coqdocvar{B}) \ensuremath{\land} (\coqdocvar{Y} $\varepsilon$ \coqdocaxiom{it} \coqdocvar{B})) \ensuremath{\rightarrow} (\coqdocvar{B} $\varepsilon$ \coqdocaxiom{idt} \coqdocvar{X} \coqdocvar{Y})) \ensuremath{\lor} \coqdoceol
\coqdocindent{8.90em}
(\coqdocvar{B} $<$ \coqdocvar{A} \ensuremath{\land} \ensuremath{\forall} \coqdocvar{X} \coqdocvar{Y}, (\coqdocvar{X} $\varepsilon$ \coqdocvar{balls}) \ensuremath{\land} (\coqdocvariable{Y} $\varepsilon$ \coqdocvar{balls}) \ensuremath{\land} (\coqdocvar{B} $\varepsilon$ \coqdocaxiom{edt} \coqdocvar{X} \coqdocvar{Y}) \ensuremath{\land}  \coqdoceol
\coqdocindent{8.90em}
 (\coqdocvar{X} $\varepsilon$ \coqdocaxiom{it} \coqdocvar{A}) \ensuremath{\land}  (\coqdocvar{Y} $\varepsilon$ \coqdocaxiom{it} \coqdocvar{A}) \ensuremath{\rightarrow} (\coqdocvar{A} $\varepsilon$ \coqdocaxiom{idt} \coqdocvar{X} \coqdocvar{Y}))).\coqdoceol
\coqdocemptyline
\end{small}
\end{coqdoccode} 
Definition 6: point.
\begin{coqdoccode}
\begin{small}
\coqdocemptyline
\coqdockw{Parameter} \coqdocaxiom{POINT}   :  \ensuremath{\forall} \coqdocvar{P} \coqdocvar{B}, \coqdocvar{P} $\varepsilon$ (\coqdocaxiom{point} \coqdocvar{B}) $\equiv\equiv$ ((\coqdocvar{P} $\varepsilon$ \coqdocvar{P}) \ensuremath{\land} (\coqdocvariable{B} $\varepsilon$ \coqdocvar{balls}) \ensuremath{\land} \coqdoceol
\coqdocindent{8.90em}
 \ensuremath{\forall} \coqdocvar{B'}, (\coqdocvar{B'} $\varepsilon$ \coqdocvar{balls})  \ensuremath{\land} \coqdocvar{B'} \coqdocaxiom{con} \coqdocvar{B}).\coqdoceol
\coqdocemptyline
\end{small}
\end{coqdoccode} 
Definition 7: equidistance.
\begin{coqdoccode}
\begin{small}
\coqdocemptyline
\coqdockw{Parameter} \coqdocaxiom{EQUID}    : \ensuremath{\forall} \coqdocvar{A} \coqdocvar{B} \coqdocvar{C}, \coqdocvar{A} $\varepsilon$ \coqdocaxiom{equid} \coqdocvar{B} \coqdocvar{C} $\equiv\equiv$ ((\coqdocvar{A} $\varepsilon$ \coqdocvar{balls}) \ensuremath{\land} (\coqdocvar{B} $\varepsilon$ \coqdocvar{balls}) \ensuremath{\land} \coqdoceol
\coqdocindent{8.90em} (\coqdocvar{C} $\varepsilon$ \coqdocvar{balls}) \ensuremath{\land} \ensuremath{\exists} \coqdocvar{X}, ((\coqdocvar{X} $\varepsilon$ \coqdocvar{balls}) \ensuremath{\land} (\coqdocvar{X} $\varepsilon$ \coqdocaxiom{con} \coqdocvar{A}) \ensuremath{\land} \ensuremath{\lnot} \ensuremath{\exists} \coqdocvar{Y}, \coqdoceol
\coqdocindent{8.90em}
 ((\coqdocvar{Y} $\varepsilon$ \coqdocvar{balls}) \ensuremath{\land} \coqdocvar{Y} $\varepsilon$ (\coqdocaxiom{union} \coqdocvar{B} \coqdocvar{C}) \ensuremath{\land} (\coqdocvar{Y} $\leqslant$ \coqdocvar{X}) \ensuremath{\lor} (\coqdocvar{Y} $\varepsilon$ \coqdocaxiom{ext} \coqdocvar{X})))).\coqdoceol
\coqdocemptyline
\end{small}
\end{coqdoccode} 
Definition 8: solids.\label{solid}
\begin{coqdoccode}
\begin{small}
\coqdocemptyline
\coqdockw{Parameter} \coqdocaxiom{TarskiD8}  : \ensuremath{\forall} \coqdocvar{A}, \coqdocvar{A} $\varepsilon$ \coqdocvar{solids} $\equiv\equiv$ \ensuremath{\exists} \coqdocvar{B}, (\coqdocvar{B} $\varepsilon$ \coqdocvar{B} \ensuremath{\land} (\coqdocvar{B} $\varepsilon$ \coqdocaxiom{coll} \coqdocvar{balls}) \ensuremath{\land}
(\coqdocvar{A} $\varepsilon$ \coqdocaxiom{subcoll} \coqdocvar{B})).  
\coqdocemptyline
\end{small}
\end{coqdoccode} 
Definition 9: interior point.
\begin{coqdoccode}
\begin{small}
\coqdocemptyline
\coqdockw{Parameter} \coqdocaxiom{IPOINT}  : \ensuremath{\forall} \coqdocvar{P} \coqdocvar{X} \coqdocvar{C}, \coqdocvar{P} $\varepsilon$ (\coqdocaxiom{ipoint} \coqdocvar{X}) $\equiv\equiv$ (\coqdocvar{X} $\varepsilon$ \coqdocvar{solids} \ensuremath{\land} \coqdocvar{P} $\varepsilon$ (\coqdocaxiom{point} \coqdocvar{C}) \ensuremath{\land} \ensuremath{\exists} \coqdocvar{A'}, ((\coqdocvar{A'} $\varepsilon$ \coqdocvar{balls}) \ensuremath{\land} \coqdoceol
\coqdocindent{8.50cm}
(\coqdocvar{A'} $\varepsilon$ \coqdocvar{P}) \ensuremath{\land} (\coqdocvar{A'} $\leqslant$ \coqdocvar{X}))).\coqdoceol
\coqdocemptyline
\end{small}
\end{coqdoccode} 
\subsection{Revisiting the Axiom System}
The axiom system of Tarski for the geometry of solid can be broadly divided in three parts, (i) axioms stating the existence of a correspondence between notions of the geometry of solids and notions of ordinary point geometry, (ii) two axioms establishing a correspondence between notions of the geometry of solids and topology and (iii) internal axioms that are derivable from Le{\'s}niewski's mereology. Axioms of the former part are:
\begin{ax}
The notions of point and equidistance of two points to a third satisfy all axioms of ordinary Euclidean geometry of three dimensions.
\end{ax} 
More specifically, this axiom states that (i) points as they are introduced in definition 6 correspond to points of an ordinary point-based geometry and (ii) the relation \begin{coqdoccode}\begin{small}\coqdocaxiom{EQUID}\end{small}
\end{coqdoccode} corresponds to an ordinary equidistance relation. With $\Pi$ standing for mereogeometrical points, the structure $\langle \Pi, \: \mathsf{EQUID}\rangle$ is a Pieri's structure \cite{Tarski29}. Then, it can be proved that $\langle \Pi, \: \mathsf{EQUID}\rangle$ is isomorphic to ordinary Euclidian geometry $\langle \mathbb{R}^3, \: \mathsf{EQUID}^{\mathbb{R}^3}\rangle$ (see \cite{Grusz08} p 500).
\begin{ax}
If $A$ is a solid, the class $\alpha$ of all interior points of $A$ is a non-empty regular open set.
\end{ax}
\begin{ax}
If the class $\alpha$ of points is a non-empty regular open set, there exists a solid $A$ such that $\alpha$ is the class of all its interior points.
\end{ax}
The second axiomatic part relies on the  structure $\langle \Pi, \: \mathsf{EQUID}\rangle$. It follows that we are able to define the family of open balls $\mathcal{O}b_\Pi$ in it and then introduce in $\Pi$ the family $\mathcal{O}_\Pi$ of open sets together with appropriate topological operations of closure and openness such that $\langle \Pi, \: \mathcal{O}_\Pi \rangle$ is a topological space. If in $\langle \Pi, \: \mathcal{O}_\Pi \rangle$, we introduce the family $\mathcal{O}r^0_\Pi$ of all regular open sets excluding the empty set, then it can be proved that $\langle \mathcal{O}r^0_\Pi, \: \mathcal{O}b_\Pi, \: \subseteq \rangle$ is isomorphic to $\langle \mathcal{O}r^0_{\mathbb{R}^3}, \: \mathcal{O}b_{\mathbb{R}^3}, \: \subseteq \rangle$ (see \cite{Grusz08} for more details).
\par The third axiomatic part of the geometry of solid is derivable from Le{\'s}niewski's mereology as follows. The first axiom of Tarski is stated as:
\begin{ax}\label{A41}
If $A$ is a ball and $B$ a part of $A$, there exists a ball $C$ which is a part of $B$.
\end{ax}
Using previous lemmas and definitions from Le{\'s}niewski's mereology, it can be proved in Coq as the theorem:
\begin{coqdoccode}
\begin{small}
\coqdocemptyline
\coqdocnoindent
\coqdockw{Theorem} \coqdocaxiom{TA4}         : \coqdockw{\ensuremath{\forall}} \coqdocvar{A} \coqdocvar{B}, (\coqdocvar{A} $\varepsilon$ \coqdocvar{balls} \ensuremath{\land} \coqdocvar{B} $\varepsilon$ \coqdocaxiom{el} \coqdocvar{A}) \ensuremath{\rightarrow} \ensuremath{\exists} \coqdocvar{C}, (\coqdocvar{C} $\varepsilon$ \coqdocvar{balls} $\rightarrow$ \coqdocvar{C} $\varepsilon$ \coqdocaxiom{el} \coqdocvar{B}). \coqdoceol
\coqdockw{Proof}.\coqdoceol
\coqdocindent{1cm}\coqdockw{intros} \coqdocvar{A} \coqdocvar{B} \coqdocvar{H1}.\coqdoceol
\coqdocindent{1cm}\coqdockw{destruct} \coqdocvar{H1} \coqdockw{as} \coqdockw{[}\coqdocvar{H1} \coqdocvar{H2}\coqdockw{]}.\coqdoceol
\coqdocindent{1cm}\coqdockw{assert} (\coqdocvar{H0}:=\coqdocvar{H1});\coqdockw{apply} \coqdocaxiom{OntoT5} \coqdockw{in} \coqdocvar{H0}.\coqdoceol
\coqdocindent{1cm}\coqdockw{apply} \coqdocaxiom{XIII} \coqdockw{in} \coqdocvar{H1}.\coqdoceol
\coqdocindent{1cm}\coqdockw{apply} \coqdocaxiom{MereoT44} \coqdockw{in} \coqdocvar{H2}.\coqdoceol
\coqdocindent{1cm}\coqdockw{assert} (\coqdocvar{H3}:(\coqdocvar{B} $\varepsilon$ \coqdocvar{solids})).\coqdoceol
\coqdocindent{1cm}\coqdockw{apply} \coqdocaxiom{TarskiD8};\coqdockw{exists} \coqdocvar{A}.\coqdoceol
\coqdocindent{1cm}\coqdockw{split};\coqdockw{[} \coqdockw{assumption} $|$ \coqdockw{split};\coqdockw{assumption}\coqdockw{]}.\coqdoceol
\coqdocindent{1cm}\coqdockw{clear} \coqdocvar{H0} \coqdocvar{H1} \coqdocvar{H2}.\coqdoceol
\coqdocindent{1cm}\coqdockw{assert} (\coqdocvar{H1}:=\coqdocvar{H3});\coqdockw{apply} \coqdocaxiom{OntoT5} \coqdockw{in} \coqdocvar{H3}.\coqdoceol
\coqdocindent{1cm}\coqdockw{apply} \coqdocaxiom{TarskiD8} \coqdockw{in} \coqdocvar{H1}.\coqdoceol
\coqdocindent{1cm}\coqdockw{destruct} \coqdocvar{H1} \coqdockw{as} \coqdockw{[}\coqdocvar{C} \coqdocvar{H1}\coqdockw{]}.\coqdoceol
\coqdocindent{1cm}\coqdockw{decompose} \coqdockw{[}\coqdockw{and}\coqdockw{]} \coqdocvar{H1};\coqdockw{clear} \coqdocvar{H1}.\coqdoceol
\coqdocindent{1cm}\coqdockw{apply} \coqdocaxiom{MD3} \coqdockw{in} \coqdocvar{H2};\coqdockw{apply} \coqdocaxiom{MereoT44} \coqdockw{in} \coqdocvar{H4}.\coqdoceol
\coqdocindent{1cm}\coqdockw{destruct} \coqdocvar{H2} \coqdockw{as} \coqdockw{[}\coqdocvar{H0} \coqdocvar{H2}\coqdockw{]};\coqdockw{clear} \coqdocvar{H0}.\coqdoceol 
\coqdocindent{1cm}\coqdockw{apply} \coqdocvar{H2} \coqdockw{in} \coqdocvar{H4};\coqdockw{clear} \coqdocvar{H2}.\coqdoceol
\coqdocindent{1cm}\coqdockw{destruct} \coqdocvar{H4} \coqdockw{as} \coqdockw{[}\coqdocvar{E} \coqdocvar{H4}\coqdockw{]};\coqdockw{destruct} \coqdocvar{H4} \coqdockw{as} \coqdockw{[}\coqdocvar{F} \coqdocvar{H4}\coqdockw{]}.\coqdoceol
\coqdocindent{1cm}\coqdockw{decompose} \coqdockw{[and]} \coqdocvar{H4};\coqdockw{clear} \coqdocvar{H4}.\coqdoceol
\coqdocindent{1cm}\coqdockw{exists} \coqdocvar{F};\coqdockw{intro};\coqdockw{assumption}.\coqdoceol
\coqdockw{Qed.}
\coqdoceol
\coqdocemptyline
\end{small}
\end{coqdoccode} 
The second axiom says that:
\begin{ax}\label{A42}
If $A$ is a solid and $B$ a part of $A$, then $B$ is also a solid.
\end{ax}
The proof in Coq is detailed below.
\begin{coqdoccode}
\begin{small}
\coqdocemptyline
\coqdocnoindent
\coqdockw{Theorem} \coqdocaxiom{TA4'}         : \coqdockw{\ensuremath{\forall}} \coqdocvar{A} \coqdocvar{B}, (\coqdocvar{A} $\varepsilon$ \coqdocvar{solids} \ensuremath{\land} \coqdocvar{B} $\varepsilon$ \coqdocaxiom{el} \coqdocvar{A})\ensuremath{\rightarrow} \coqdocvar{B} $\varepsilon$ \coqdocvar{solids}.\coqdoceol
\coqdockw{Proof}.\coqdoceol
\coqdocindent{1cm}\coqdockw{intros} \coqdocvar{A} \coqdocvar{B} \coqdocvar{H}.\coqdoceol
\coqdocindent{1cm}\coqdockw{destruct} \coqdocvar{H} \coqdockw{as} \coqdockw{[}\coqdocvar{H1} \coqdocvar{H2}\coqdockw{]}.\coqdoceol
\coqdocindent{1cm}\coqdockw{apply} \coqdocaxiom{TarskiD8} \coqdockw{in} \coqdocvar{H1}.\coqdoceol
\coqdocindent{1cm}\coqdockw{destruct} \coqdocvar{H1} \coqdockw{as} \coqdockw{[}\coqdocvar{C} \coqdocvar{H3}\coqdockw{]}.\coqdoceol
\coqdocindent{1cm}\coqdockw{destruct} \coqdocvar{H3} \coqdockw{as} \coqdockw{[}\coqdocvar{H3} \coqdocvar{H4}\coqdockw{]};\coqdockw{destruct} \coqdocvar{H4} \coqdockw{as} \coqdockw{[}\coqdocvar{H4} \coqdocvar{H5}\coqdockw{]}.\coqdoceol
\coqdocindent{1cm}\coqdockw{apply} \coqdocaxiom{MereoT44} \coqdockw{in} \coqdocvar{H2}.\coqdoceol
\coqdocindent{1cm}\coqdockw{apply} \coqdocaxiom{TarskiD8};\coqdockw{exists} \coqdocvar{C}.\coqdoceol
\coqdocindent{1cm}\coqdockw{split};\coqdockw{[} \coqdockw{assumption} $|$ \coqdockw{split};\coqdockw{[} \coqdockw{assumption} $|$ \coqdoceol
\coqdocindent{2cm}\coqdockw{apply} (\coqdocaxiom{Transitive\_subcoll} \coqdocvar{B} \coqdocvar{A} \coqdocvar{C});\coqdockw{split};\coqdockw{assumption} \coqdoceol
\coqdocindent{2cm}\coqdockw{]].}
\coqdockw{Qed.}
\coqdoceol
\coqdocemptyline
\end{small}
\end{coqdoccode}
Another axiom is provided by Tarski which relies on the definition of \textit{interior points} and is formulated as follows.
\begin{ax}\label{A43}
If $A$ and $B$ are solids, and all the interior points of $A$ are at the same time interior points of $B$, then $A$ is a part of $B$.
\end{ax}
As advocated by the author, axiom \ref{A43} (i) relies on the interplay between \textit{interior points} and the set-based definition of regular open sets, and thus requires to state the relation between a solid and its interior points and (ii) is merely an alternative of any axiom among \ref{A41} and \ref{A42}. It is well-known that the more axioms one assumes in a formal system, the harder it becomes to preserves its soundness.
If we use one of the axioms \ref{A41} or \ref{A42} we get a more stronger system than the one obtained with axiom \ref{A43} (see theorem 1.3 in \cite{Grusz08}). It follows that the resulting axiom system obtained by blurring axiom \ref{A43} provides a minimal system that can serve as a basis for constructing spatial theories.
\section{Conclusion}\label{conclusion}
Generalizing solids to spatial regions, geometrical theories based on mereology present an appealing impact on spatial theories. As underlined in \cite{Borgo10}, they provide formal theories adequate for different tasks. Among their benefits, (i) they make possible a direct mapping from empirical entities and laws to theoretical entities and formulas, (ii) they have the ability to formalize human learning, conceptualization, and categorization of spatial entities and relations and (iii) they have received a particular emphasis in the field of formal ontology with mereogeometrical notions. The theory of Tarski, has been proved to be semantically complete with regards to the models expressed in terms of $R^n$ and has been axiomatized by Bennett \cite{Benn01}. 
\par Major problems are that (i) the set-based interpretation (e.g., interpreting \begin{coqdoccode}\begin{small}\coqdocvar{A} $\varepsilon$ \coqdocaxiom{pt} \coqdocvar{B}\end{small}\end{coqdoccode} as $A \subseteq B$) considerably weakens the logical power of Le{\'s}niewski's framework and (ii) the approach described in \cite{Grusz08} in which the fourth axiom is replaced with a new postulate asserting that the domain of discourse of Tarski's theory coincides with arbitrary mereological sums of balls, does not simplify the work of Tarski either. What we have proposed so far to avoid these problems, is a logical foundation having the following properties: (i) the proposed set of structures featuring geometrical entities and relations relies on Tarski's mereogeometry, (ii) it has a model in ordinary three-dimensional Euclidian geometry \cite{Tarski56a}, (iii) it is based on three axioms instead of four (iv) it is coherent with Le{\'s}niewski's mereology and does not suffer the defects cited in \cite{Betti12} and (v) it will serve as a basis for spatial reasoning with full compliance with Le{\'s}niewski's systems. These systems have precisely the property to be scalable, which is a significant argument for extending the theory with new definitions in appropriate applications. Future work will develop this last aspect.


\begin{thebibliography}{9}

\bibitem{Appel04}
	A.~W. Appel and A.~P. Felty,
	\textsl{Dependent types ensure partial correctness of theorem provers}.
	J. of Functional Programming, 14(1), 
	Cambridge Univ. Press, 3--19 2004
\bibitem{Asher95}
	N. Asher and L. Vieu,
	\textsl{Toward a geometry of common sense: A semantics and a complete axiomatization of mereotopology}.
	In Procs. of the Int. Joint Conf. on Artificial Intelligence (IJCAI-95),
	Montreal, 1995
\bibitem{Aurnag95}
	M. Aurnague and L. Vieu,
	\textsl{A theory of space-time for natural language semantics}.
	in: K. Korta and J. M. Larrazabal, eds., Semantics and Pragmatics of Natural Language: Logical and Computational Aspects, ILCLI Series I, Univ. Pais Vasco, San Sebastian, 69--126, 1995
\bibitem{Bart03}
	G. Barthe, V. Capretta and O. Pons,
	\textsl{Setoids in type theory}. 
	J. of Functional Programming, 13(2), 261--293, 2003
\bibitem{Benn01}	
	B. Bennett,
	\textsl{A categorical axiomatisation of region-based geometry}.
	Fundamenta Informaticae 46(1-2), 145--158, 2001
\bibitem{Bertot04}
	Y. Bertot and P. Cast{\'e}ran,   
	\textsl{Interactive Theorem Proving and Program Development. Coq'Art: The Calculus of Inductive Constructions}. 
	Texts in Theoretical Computer Science, 
	An EATCS series, Springer Verlag, 2004
\bibitem{Bertot08}
	Y. Bertot, L. Th{\'e}ry,
	\textsl{Dependent Types, Theorem Proving, and Applications for a Verifying Compiler}.
	Verified Software: Theories, Tools, Experiments,
	LNCS 4171, 173--181, 2008
\bibitem{Betti12}
	A. Betti and I. Loeb,
	\textsl{On Tarski's foundations of the geometry of solids}.
	Bulletin of Symbolic Logic, 18(2), 230--260, 2012
\bibitem{Betti13}
	A. Betti,
	\textsl{Le{\'s}niewski, Tarski and the Axioms of Mereology}.
	chap. 11 in "The History and Philosophy of Polish Logic: Essays in Honour of Jan Wole{\'n}ski",
	K. Mulligan-et-al eds., 242--258, 2013		
\bibitem{Borgo96}
	S. Borgo, N. Guarino and C. Masolo,
	\textsl{A pointless theory of space based on strong congruence and connection.}	
	Procs. of 5th Int. Conf. on Principle of Knowledge Representation and Reasoning (KR'96), Morgan Kaufmann, 220--229, 1996 		
\bibitem{Borgo10}
	S. Borgo and C. Masolo,
	\textsl{Full mereogeometries.}	
	Review of Symbolic Logic, 3(4), 521--567, 2010 	
\bibitem{Church40}
	A. Church,
	\textsl{A formulation of the simple theory of types}.
	Journal of Symbolic Logic, 5, 56--68, 1940
\bibitem{Clay68}	
	R.E. Clay, 
	\textsl{The consistency of Le{\'s}niewski's mereology relative to the real number system}.
	The Journal of Symbolic Logic, 33, 251--257, 1968
\bibitem{Clay74}
	R. E. Clay
	\textsl{Relation of Le{\'s}niewski's Mereology to Boolean Algebra}.
	The Journal of Symbolic Logic
	39(4), 638--648, 1974
\bibitem{Coq88}
	T. Coquand and G. Huet,
	\textsl{The calculus of constructions}.
	Information and Computation,
	76 (2--3), 95--120, 1988
\bibitem{Coq90}	
	T. Coquand and C. Paulin-Mohring.
	\textsl{Inductively defined types}. 
	In P. Martin-L{\"o}f and G. Mints, editors,
	Procs of Colog'88, LNCS 417, Springer-Verlag, 1990
\bibitem{Dapoi15}
	R.Dapoigny and P. barlatier,
	Tarski\_Mereogeometry,
	
\bibitem{Eschenbach95}
    C. Eschenbach and W. Heydrich,
	\textsl{Classical mereology and restricted domains}.
	International Journal of Human-Computer Studies, 43(5–6), 723--740, 1995
\bibitem{Gerla95}	
	G. Gerla,
	\textsl{Pointless Geometries}. 
	in F. Buekenhout, eds., Handbook of Incidence Geometry, 
	Elsevier, 1015--1031, 1995
\bibitem{Gessler05}
	N. Gessler,
	\textsl{Introduction {\`a} l'oeuvre de S. Le{\'s}niewski. part III: La m{\'e}r{\'e}ologie}.
	CdRS, Universit{\'e} de Neuch{\^a}tel, 2005
\bibitem{Grusz08}
	R. Gruszczy{\`n}ski and R. Pietruszczak,
	\textsl{Full Development of Tarski's Geometry of Solids}.
	The Bulletin of Symbolic Logic 14(4), 481--540, 2008
\bibitem{Hahmann10}
	T. Hahmann and M. Gr{\"u}ninger,
	\textsl{Region-Based Theories of Space: Mereotopology and Beyond}.
	in Qualitative Spatio-Temporal Representation and Reasoning: Trends and Future Directions, Shyamanta M. Hazarika Ed., 1--62, 2012
\bibitem{Henkin63}
	L. Henkin,
	\textsl{A theory of propositional types}.
	Fundamenta Mathematicae, 52 
	(1963) 323--334. Errata, 53, 119.
\bibitem{Hof95}
	M. Hofmann,
	\textsl{Extensional concepts in intensional type theory}.
	Phd thesis, University of Edinburgh, 1995
\bibitem{Korten98}
	D. Kortenkamp, R. P. Bonasso and R. Murphy (eds.),
	\textsl{Artificial Intelligence and Mobile Robots: case studies of successful robot systems}.
	MIT Press, 1998
\bibitem{Kuipers87}
	B. J. Kuipers and Y. T. Byun,
	\textsl{A qualitative approach to robot exploration and map learning}. 
	in:	Procs. of the IEEE Workshop on Spatial Reasoning and Multi-Sensor Fusion,
	Morgan Kaufmann, San Mateo CA, 390--404, 1987
\bibitem{Lejewski69}
	C. Lejewski,
	\textsl{Consistency of Le{\'s}niewski's Mereology}.
	Journal of Symbolic Logic, 34 (3), 321--328, 1969		
\bibitem{Lesniewski16}
	S. Le{\'s}niewski,	\textsl{Podstawy og{\'o}lnej teoryi mnogosci. I},
	Moskow: Prace Polskiego Kola Naukowego w Moskwie, Sekcya matematyczno-przyrodnicza, 1916 (English translation by D. I. Barnett: \textsl{Foundations of the General Theory of Sets. I}, in S. Le{\'s}niewski, Collected Works, ed. S. J. Surma, J. Srzednicki, D. I. Barnett, and F. V. Rickey, Dordrecht: Kluwer, 1, 129--173, 1992
\bibitem{Lesniewski38}
	\textsl{S. Le{\'s}niewski: Einleitende Bemerkungen zur Fortsezung meiner
	Miteilung u.d.T. "Grundz{\"u}ge eines neuen Systems der Grundlagen
	der Mathematik"}. Collectanea Logica, vol. I, 1--60, 1938	
\bibitem{Oury05}
	N. Oury, 
	\textsl{Extensionality in the Calculus of Constructions}.
	Procs. of TPHOL'05, 
	LNCS 3603, 278--293, 2005
\bibitem{Quine56}	
	W. Quine,
	\textsl{Unification of universes in set theory}.
	J. of Symb. Logic, 21, 216, 1956
\bibitem{Rand92}
	D.A. Randell, Z. Cui, and A.G. Cohn,
	\textsl{A spatial logic based on regions and connection}.
	In Proc. 3rd Int. Conf. on Knowledge Representation and Reasoning, 
	165--176, San Mateo, Morgan Kaufmann, 1992
\bibitem{Rickey77}
	V.F. Rickey,
	\textsl{A Survey of Le{\'s}niewski's Logic}.
	Studia Logica
	36, 407--426, 1977
\bibitem{Seldin01}	
	J.P. Seldin,
	\textsl{Extensional Set Equality in the Calculus of Constructions}.
	J. Log. Comput. 11(3), 483--493, 2001
\bibitem{Simon87}
	P. Simons,
	\textsl{Parts: A Study in Ontology}.
	Clarendon Press, Oxford, 1987
\bibitem{Simon98}
	P. Simons,
	\textsl{Nominalism in Poland}.
	in Le{\'s}niewski's Systems Protothetic,
	Nijhoff International Philosophy Series, 
	54, 1--22, 1998
\bibitem{Sinisi83}	
	V.F. Sinisi,
	\textsl{Le{\'s}niewski's foundations of mathematics}.
	Topoi, 2(1), 3--52, 1983
\bibitem{Slupecki53}
	J. Slupecki,
	\textsl{S. Le{\'s}niewski's protothetics}. 
	Studia Logica, 1,
	44--112, 1953
\bibitem{Sobo60}	
	B. Soboci{\'n}ski,
	\textsl{On the single axioms of protothetic I}.
	Notre-Dame Journal of Formal Logic, 1, 52--73, 1960
\bibitem{Sozeau09}	
	 M. Sozeau,
	 \textsl{A New Look at Generalized Rewriting in Type Theory},
	 Journal of Formalized Reasoning 2(1), 41--62, 2009
\bibitem{Tarski23}
	 A. Tarski, 
	 \textsl{Sur le terme primitif de la Logistique}. 
	 Fundamenta Mathematicae, 4, 196--200, 1923
\bibitem{Tarski29}
	A. Tarski,
	\textsl{Les fondements de la g{\'e}om{\'e}trie des corps (Foundations of the geometry of solids)}.
	in Ksiega Pamiatkowa Pierwszego Polskiego Zjazdu Matematycznego,
	7, 29--33, 1929
\bibitem{Tarski56a}
	A. Tarski,	
	\textsl{Foundations of the geometry of solids}.
	in Logics, Semantics, Metamathematics. 
	Papers from 1923-1938 by Alfred	Tarski, Clarendon Press, 1956	 
\bibitem{Tarski56b}
	A. Tarski,
	\textsl{On the foundation of Boolean algebra}.
	Logic, Semantics, Metamathematics: Papers from 1923 to 1938,
	Clarendon Press, Oxford, 320--341, 1956
\bibitem{Varzi96}
	A.~C. Varzi,
	\textsl{Parts, wholes, and part-whole relations: The prospects of mereotopology}.
	Data and Knowledge Engineering, 20(3), 259--286, 1996

\end{thebibliography}
\end{document}